%% file: main.tex
\documentclass[conference,a4paper]{IEEEtran}

\usepackage[utf8]{inputenc} 
\usepackage[T1]{fontenc}
\usepackage{url}              
\usepackage{cite}             

\usepackage[cmex10]{amsmath}  
\usepackage{mleftright}       
\mleftright                   

\usepackage{graphicx}         
\usepackage{booktabs}         

\usepackage{algorithmicx}     

\usepackage[caption=false,font=footnotesize]{subfig}

\usepackage{amsfonts,amssymb,dsfont}
\usepackage{bm}
\usepackage{amsthm}
\usepackage{stfloats}
\usepackage{xcolor}
\usepackage{notations}
\usepackage{comment}

\newcommand{\teng}[1]{{\color{red} #1}}

\begin{document}

\title{Mismatched estimation of non-symmetric \\ rank-one matrices corrupted by structured noise} 


%
\author{%
  \IEEEauthorblockN{Teng Fu\IEEEauthorrefmark{1},
                    YuHao Liu\IEEEauthorrefmark{1},
                    Jean Barbier,
                    Marco Mondelli\IEEEauthorrefmark{2},
                    ShanSuo Liang\IEEEauthorrefmark{3}
                    and TianQi Hou\IEEEauthorrefmark{3}}
  \IEEEauthorblockA{\IEEEauthorrefmark{1}%
                    Department of Mathematical Sciences, Tsinghua University, Beijing, China}
  \IEEEauthorblockA{\IEEEauthorrefmark{2}%
                    Institute of Science and Technology, Austria}
  \IEEEauthorblockA{\IEEEauthorrefmark{3}%
                    Theory Lab, Central Research Institute, 2012 Labs, Huawei Technologies Co., Ltd.}
  \IEEEauthorblockA{%
    Emails: \{fut21, yh-liu21\}@mails.tsinghua.edu.cn, jean.barbier.cs@gmail.com, marco.mondelli@ist.ac.at,\\
    liang.shansuo@huawei.com, thou@connect.ust.hk}
}

\maketitle

\begin{abstract}
    We study the performance of a Bayesian statistician who estimates a rank-one signal corrupted by non-symmetric rotationally invariant noise with a generic distribution of singular values. As the signal-to-noise ratio and the noise structure are unknown, a Gaussian setup is incorrectly assumed. We derive the exact analytic expression for the error of the mismatched Bayes estimator and also provide the analysis of an approximate message passing (AMP) algorithm. The first result exploits the asymptotic behavior of spherical integrals for rectangular matrices and of low-rank matrix perturbations; the second one relies on the design and analysis of an auxiliary AMP. The numerical experiments show that there is a performance gap between the AMP and Bayes estimators, which is due to the incorrect estimation of the signal norm.
\end{abstract}

\section{Introduction and set-up}
\label{sec:introduction}

The estimation of low-rank matrices from noisy observations has found numerous applications in statistics and machine learning: sparse principal component analysis (PCA) \cite{johnstone2009consistency, zou2006sparse}, community detection \cite{abbe2017community, moore2017computer} or group synchronization \cite{perry2018message} just to cite a few. Studies in statistics \cite{johnstone2001distribution, paul2007asymptotics} and random matrix theory \cite{baik2005phase, baik2006eigenvalues, bai2012sample, benaych2011eigenvalues, benaych2012singular} have established the properties of spectral algorithms. In parallel, the application of approximate message passing (AMP) algorithms for inference, whose asymptotic performance can be tracked by state evolution (SE) \cite{bayati2011dynamics, bolthausen2014iterative}, has flourished: they have been used for (generalized) linear models \cite{bayati2011dynamics, donoho2009message,barbier2019optimal, maillard2020phase, mondelli2021approximate, rangan2011generalized} or, closer to our setting, low-rank matrix recovery with Gaussian noise \cite{deshpande2014information,macris2020all, dia2016mutual,fletcher2018iterative, montanari2021estimation}.

 
 Recently, the interest towards more ``realistic'', richer models of matrix estimation, which include mismatch (in contrast to the Bayes-optimal setting generally studied) and/or statistical dependencies in the noise (rather than i.i.d. Gaussian) has strongly increased, see \cite{pourkamali2022mismatchedsym, pourkamali2022mismatchednonsym,camilli2022inference, barbier2022price}. The present paper extends \cite{barbier2022price} to the case of non-symmetric signal and noise matrices. Our main contributions are \emph{(i)} an expression for the mean-square error (MSE) of the mismatched Bayes estimator, and \emph{(ii)} a state evolution for the Gaussian AMP in the mismatched noise-statistics case. Comparing the Bayesian, AMP and spectral estimators, we find a (surprising) phenomenology similar to the symmetric case \cite{barbier2022price}; in particular the Bayes and AMP estimators do not match as one could expect. 

\label{sec:problem-model}
We now introduce the set-up. Let $\bm{u}^*\in \mathbb{R}^n$ and $\bm{v}^*\in \mathbb{R}^m$ be uniformly distributed on the sphere $\mathbb{S}^{n-1}(\sqrt{n})$ and $\mathbb{S}^{m-1}(\sqrt{m})$, respectively. The uniform measure on the sphere $\mathbb{S}^{n-1}(\sqrt{n})$ will be denoted by $P_n$. We focus on the task of inferring $(\bu^*,\bv^{*})$ from observed data $\bm{Y}$ constructed as
\begin{equation}
    \bm{Y}=\sqrt{\frac{\lambda_*}{mn}}\bm{u}^*\bm{v}^{*\mathsf{T}}+\bm{Z}\in \mathbb{R}^{n\times m}.\label{data}
\end{equation}
The scalar $\lambda_*\ge 0$ is the signal-to-noise ratio (SNR). We assume that $n/m\to\alpha\in(0, 1]$ as $n\to\infty$ ($\alpha\le 1$ without loss of generality since we can consider $\bm{Y}^\mathsf{T}$ when $\alpha>1$). With a slight abuse of notation, we will also denote $\alpha=\alpha_n=n/m$. The matrix $\bm{Z}$ representing the noise is bi-unitary invariant, meaning that $\bm{U}\bm{Z}\bm{V}^\mathsf{T}$ equals in law $\bm{Z}$ for any orthogonal matrices $\bm{U}\in\mathbb{R}^{n\times n}$ and $\bm{V}\in\mathbb{R}^{m\times m}$. Its empirical singular law $\mu_{\bm{Z}}=\frac{1}{n}\sum_{i\le n}\delta_{\sigma_i}$, where $(\sigma_i)_{i\le n}$ are the singular values of $\bm{Z}$, converges in the weak sense to $\mu_{Z}$ as $n\to\infty$. The asymptotic empirical spectral density of $\bm{Z}\bm{Z}^\mathsf{T}$ is denoted by $\rho_Z$ and is obtained from $\mu_{Z}$ through the change of density $\mu_Z(x)\dif x=\rho_Z(y)\dif y$ with $y=x^2$. We also know that the asymptotic spectral density of $\bm{Z}^\mathsf{T}\bm{Z}$ satisfies $\hat{\rho}_Z=\alpha\rho_Z+(1-\alpha)\delta_0$. Note that $\bm{Y}$ has the same limiting singular law as $\bm{Z}$, but not necessarily the same limiting largest singular value, so we will later omit the subscript $Z$ without causing ambiguity for the asymptotic law of both $\bZ$ and $\bY$, e.g., $\mu$ for both $\mu_Z=\mu_Y$ and $\rho$ for $\rho_Z=\rho_Y$. We denote by $\bar{\nu}$ the limit of the largest singular value of $\bm{Y}$ and by $\bar{\gamma}$ the supremum of the support of $\mu$. $\bZ$ has asymptotically no outliers, so the end-point $\bar{\gamma}$ of the support of $\mu$ is also its asymptotic largest singular value.

For a sequence of estimators $M(\bm{Y})=M_n(\bm{Y})\in\mathbb{R}^{n\times{m}}$ of $\bm{u}^*\bm{v}^{*\mathsf{T}}$, we define the corresponding MSE as
\begin{equation}
    \text{MSE}_n(M):=\frac{1}{2mn}\mathbb{E}\|\bm{u}^*\bm{v}^{*\mathsf{T}}-M(\bm{Y})\|_{\text{F}}^2,
    \label{eq:MSE-n}
\end{equation}
where $\mathbb{E}$ denotes the expectation over $(\bm{Z},\bm{u}^*,\bm{v}^*)$. We also consider another performance measure which is insensitive to the norm of estimators $\bm{u}$, $\bm{v}$ of $\bm{u}^*$, $\bm{v}^*$, namely, the rescaled overlap (below $\langle \,\cdot\,,\, \cdot\,\rangle$ is the standard inner product):
\begin{equation}
    \text{Overlap}(\bm{u},\bm{v}):=\lim_{n\to\infty}\frac{|\langle\bm{u},\bm{u}^*\rangle|}{\|\bm{u}\|\|\bm{u}^*\|}\frac{|\langle\bm{v},\bm{v}^*\rangle|}{\|\bm{v}\|\|\bm{v}^*\|}. \label{eq:overlap-n}
\end{equation}
If the estimators' norm vanishes, we set the overlap to zero.

\section{Mismatched Bayes estimator}

We consider the performance of a statistician who is wrongly assuming a Gaussian noise with i.i.d. $\mathcal{N}(0,1/m)$ entries and an SNR $\lambda$ possibly different from the true value $\lambda_*$. The mismatched posterior distribution used for inference therefore reads
\begin{equation}
    P_{\text{mis}}(\d \bm{u}, \d \bm{v}\mid\bm{Y})=\frac{1}{Z_n(\bm{Y})}e^{\sqrt{\frac{\lambda}{\alpha}}\langle \bm{u}, \bm{Y}\bm{v} \rangle}P_n(\d \bm{u})P_m(\d \bm{v})\label{posterior}
\end{equation}
and the partition function (i.e., posterior normalization) is
\begin{equation}
    Z_n(\bm{Y}) := \int P_n(\d\bm{u})P_m(\d\bm{v})\exp\Big( \sqrt{\frac{\lambda}{\alpha}}\langle \bm{u}, \bm{Y}\bm{v} \rangle \Big). \label{equ-Zn}
\end{equation}
The associated mismatched Bayes estimator we analyze is 
\begin{equation}
    M_{\text{mis}}(\bm{Y}):=\int\bm{u}\bm{v}^\mathsf{T}{P}_{\text{mis}}(\d\bm{u},\d\bm{v}\mid\bm{Y}).
\end{equation}
We will analyze the $n\to\infty$ limit of \eqref{eq:MSE-n}, denoted by $\text{MSE}(M_{\text{mis}})$. However, definition  \eqref{eq:overlap-n} is not meaningful for the Bayes case (the posterior mean of $\bu$, $\bv$ can be $0$ by sign symmetry). Thus, in the Bayesian case, the overlap is defined as
\begin{equation}
\begin{IEEEeqnarraybox}{rl}
    &\text{Overlap}_{\text{mis}}:=\lim_{n\to\infty}\Big(\frac{1}{mn}\frac{1}{\|M_{\text{mis}}(\bm{Y})\|_{\text{F}}^2} \\
    &\quad\times \int P_{\text{mis}}(\d\bm{u},\d\bm{v}\mid\bm{Y})\langle\bm{u},\bm{u}^*\rangle^2\langle\bm{v},\bm{v}^*\rangle^2\Big)^{{1}/{2}}.
\end{IEEEeqnarraybox}
\label{eq:overlap-mis}
\end{equation}

\noindent \emph{\bf{Definitions.}} We provide here some useful definitions and transforms directly borrowed from \cite{benaych2011rectangular,benaych2012singular}. For any symmetric probability measure $\mu$, the D-transform with ratio $\alpha$ of $\mu$ is defined as
\begin{equation*}
    D_{\mu}^{(\alpha)}(z):=\int\mu(\d t)\frac{z}{z^2-t^2} \times \Big[ \alpha\int\mu(\d t)\frac{z}{z^2-t^2} + \frac{1-\alpha}{z} \Big]
\end{equation*}
for $z>\bar{\gamma}$, and $D_{\mu}^{(\alpha),-1}(z)$ will denote its functional inverse on $(\bar{\gamma},+\infty)$. Let $T^{(\alpha)}(z):=(\alpha z+1)(z+1)$ whose inverse from $[-1,\infty)$ to $[0,\infty)$ is denoted by $T^{(\alpha),-1}(z)$, since it is increasing in this interval. Finally, the \emph{rectangular R-transform with ratio $\alpha$ of $\mu$} is defined as $C^{(\alpha)}_{\mu}(z)={T^{(\alpha),-1}}(z(D_{\mu}^{(\alpha),-1}(z))^2)$ for $z\neq 0$ and $C^{(\alpha)}_{\mu}(0)=0$. We sometimes omit the superscript $\alpha$ without causing ambiguity, e.g., $D_{\mu}^{-1}(z)$ for $D_{\mu}^{(\alpha),-1}(z)$.

\subsection{Log-partition function}
Our first result concerns the log-partition $\ln  Z_n(\bm{Y})$, whose expression allows to locate the phase transitions of inference and derive the MSE of the Bayes estimator. Our calculations take inspiration from \cite{pourkamali2022mismatchedsym,barbier2022price,pourkamali2022mismatchednonsym}. The main point is that thanks to the spherical nature of the prior distribution of the signals $(\bu^*,\bv^*)$, the partition function \eqref{equ-Zn} is a \emph{rectangular spherical integral} $I_n$ \cite{benaych2011rectangular}. In fact, we can identify the partition function as $Z_n(\bm{Y})=I_n(\sqrt{\lambda/\alpha},\bm{Y})$, where
\begin{equation}
    I_n(\theta,\bm{Y}):=\int P_n(\d\bm{u})P_m(\d\bm{v})\exp\left( \theta\langle \bm{u}, \bm{Y}\bm{v} \rangle \right). \label{spinglass}
\end{equation}
The rigorous asymptotic formula for $I_n$ of \cite{benaych2011rectangular} is restricted to small values of $\theta$. In contrast, the non-rigorous formula shown in \cite{kabashima2008inference,maillard2019high} applies to any $\theta$ and reads, as $n\to\infty$,
\begin{equation}
    \frac{1}{n}\ln I_n(\theta,\bm{Y})\to  \frac{1}{2\alpha}\,\underset{(z_1, z_2)\in \mathcal{D}(\theta,\bar\nu)}{\rm extr}\, \phi(z_1, z_2)-\frac{1+\alpha}{2\alpha},  \label{Kaba-int}
\end{equation}
where $\phi=\phi(z_1, z_2)$ reads
\begin{equation}
    \phi:= -\int\hat{\rho}(\d t)\ln(z_1z_2-\theta^2t)+z_1\alpha+z_2-(\alpha-1)\ln z_1 \label{func-phi12}
\end{equation}
and $\mathcal{D}(\theta,\bar\nu):=\{(z_1,z_2) \in \mathbb{R}_{>0}\times\mathbb{R}_{>0}: z_1z_2\ge \theta^2\bar \nu^2\}$. The extremum ${\rm extr}\, \phi$ means that $\phi$ is evaluated at $(z_1,z_2)$ verifying $\nabla \phi=\boldsymbol{0}$, if this stationary point $(z_1,z_2)$ (which can be shown to be unique) exists. If $\nabla \phi=\boldsymbol{0}$ does not possess a solution, the extremum selects the unique boundary of $\mathcal{D}$. 

We aim for a more explicit expression for (\ref{Kaba-int}). The key point is to analyze the ``high and low temperatures'' regimes in  $\theta=\sqrt{\lambda/\alpha}$ of the rectangular spherical integral \eqref{spinglass} and locate the transition point between the two. The high temperature regime corresponds to values of $\theta$ such that $\nabla \phi=\boldsymbol{0}$ does possess a solution; in the low temperature it does not and $(z_1,z_2)$ stick to their boundary value. The ``sticking transition'' separating the two regimes happens at
\begin{equation}
    \bar{\theta} := \Big(\lim_{\nu \downarrow \bar{\nu}}D_{\mu}^{(\alpha)}({\nu})\Big)^{1/2}. \label{theta_trans} 
\end{equation}
Moreover, there is also a second phase transition, this time controlled by the value of $\lambda_*$, which separates a region where the largest singular value of the data $\bY$ is separated from the bulk of singular values (and therefore non-trivial estimation is possible) from a region where there is no outlier. We call it ``BBP transition'' in reference to \cite{baik2005phase}. This transition was studied for the present setting in \cite{benaych2012singular}. Therefore, depending on the values of $\lambda_*$ and $\theta=\sqrt{\lambda/\alpha}$, four possible regimes emerge (two of which yielding the same expression for the log-partition function). The rectangular spherical integral in the two temperature regimes is derived in Appendix \ref{sec:appendix-logfunc}; we recall the background on the BBP transition in Appendix \ref{sec:appendix-lowrank}.

\vspace{5pt}
\noindent \emph{\bf{Generalized observation model.}} We are going to state our first result for a slightly more general model for the data than \eqref{data}, which will later be useful to derive the MSE of the Bayes estimator. It is defined by $\bm{Y}_{\epsilon}=\bY+\sqrt{\epsilon}\,\bm{W}$, where we added a standard Wigner matrix $\bW$ to the original data \eqref{data}, with $\epsilon\ll 1$.
Then, we can define the log-partition function $\ln Z_n(\bm{Y}_{\epsilon})$ 
where $Z_n(\bm{Y}_\epsilon)$ is defined as in \eqref{equ-Zn} but for $\bm{Y}_{\epsilon}$ instead of $\bm{Y}$. This model matches the initial one \eqref{data} when $\epsilon=0$. We define $\bar{\nu}_{\epsilon}$ as the limit as $n,m\to\infty$ of the largest singular value of $\bm{Y}_{\epsilon}$; $\bar{\gamma}_{\epsilon}$ is the limit of the largest singular value of the (generalized) noise $\bm{Z}+\sqrt{\epsilon}\,\bm{W}$, whose limiting density of singular values is $\mu_\epsilon$; the limiting eigenvalue density of $\bY_\epsilon^\mathsf{T}\bY_\epsilon$ is denoted by $\hat \rho_\epsilon$, and $\bar{h}_{\epsilon}:=\lim_{x\downarrow\bar{\gamma}_{\epsilon}^2}D_{\mu_{\epsilon}}(x)$. When $\epsilon=0$ we omit the subscript $0$, i.e., $\bar{\nu}$ for $\bar{\nu}_0$, $\bar{\gamma}$ for $\bar{\gamma}_0$, etc. Our conjecture for the log-partition function of this generalized model is as follows.

\begin{conjecture}[Log-partition function] \label{Conj:LogZ}
Define the sticking transition separating the high and low temperature regimes of the rectangular spherical integral for the generalized model: $\bar{\lambda}_{\epsilon}:=\alpha\lim_{z\downarrow\bar{\nu}_{\epsilon}}D_{\mu_{\epsilon}}(z)$. We almost surely have 
\begin{equation*}
   \lim_{n\to \infty}\frac{1}{n}\ln Z_n(\bm{Y}_{\epsilon})= \lim_{n\to \infty}\frac{1}{n}\ln I_n\big(\sqrt{\frac{\lambda}{\alpha}},\bm{Y}_{\epsilon}\big)=f_{\epsilon}^{(\alpha)}(\lambda,\lambda_*),
\end{equation*}
where
\begin{equation*}
    f_{\epsilon}^{(\alpha)}(\lambda,\lambda_*):=
    \begin{cases}
    g_{\lambda,\epsilon}^{(\alpha)}((D^{-1}_{\mu_\epsilon}(\frac{1}{\lambda_*}))^2) & \text{$\bar{h}_{\epsilon}\lambda_*\geq 1\cap\lambda\lambda_*>\alpha$}, \\
    g_{\lambda,\epsilon}^{(\alpha)}(\bar{\gamma}_{\epsilon}^2) & \text{$\bar{h}_{\epsilon}\lambda_*<1\cap\lambda>\alpha\bar{h}_{\epsilon}$}, \\ 
    \int_{0}^{\sqrt{\frac{\lambda}{\alpha}}}\frac{C_{\mu_\epsilon}^{(\alpha)}(t^2)}{t}\d t & \text{otherwise}.
    \end{cases} 
\end{equation*}
with
\begin{IEEEeqnarray}{rl}
    {g_{\lambda,\epsilon}^{(\alpha)}(x)}
    :=-\frac{1}{2\alpha}&\Big[\int\hat{\rho}_{\epsilon}(\d t)\ln(x-t)-2\alpha T^{(\alpha),-1}\Big(\frac{\lambda x}{\alpha}\Big) \nn
    &+\>(\alpha-1)\ln\Big(T^{(\alpha),-1}\Big(\frac{\lambda x}{\alpha}\Big)+1\Big)+\ln\frac{\lambda}{\alpha}\Big]. \nonumber
\end{IEEEeqnarray}
Moreover, we also have that $\frac{1}{n}\mathbb{E}\ln Z_n(\bm{Y}_{\epsilon})\to f_{\epsilon}^{(\alpha)}(\lambda,\lambda_*)$.
\end{conjecture}

There is a link between the log-partition functions of the non-symmetric case with $\alpha=1$ and the symmetric case. Indeed, from \cite{barbier2022price} there is also an expression for the partition function denoted by $f_{\epsilon}^{\text{sym}}(\lambda,\lambda_*;\rho_\epsilon)$ in the symmetric case. If we set $\mu_{\epsilon}=\rho_{\epsilon}$ we get that $f_{\epsilon}^{(1)}(\lambda,\lambda_*;\mu_\epsilon) = 2f_{\epsilon}^{\text{sym}}(\lambda,\lambda_*;\rho_{\epsilon})$.
We emphasize that the noise $\bm{Z}$ in the $\alpha=1$ non-symmetric case is a Gaussian matrix instead of Wigner, and the signal vectors $\bm{u}^*\neq\bm{v}^*$. The interpretation of the coefficient $2$ in the above expression is that in the symmetric case, there is just half the ``information'' compared to the non-symmetric case because $Z_{ij}=Z_{ji}$ holds when $\bZ$ is a Wigner matrix.

We remark that in contrast to \cite{barbier2022price} for the symmetric case, the present results are not rigorous. This stems from the fact that much less is known about the rectangular version of the spherical integral \cite{benaych2011rectangular}, and more generically about rectangular free probability \cite{benaych2007infinitely} compared to free probability.

\subsection{From the log-partition function to the mean-square error}

Connecting the log-partition function to the MSE means deriving a ``generalized I-MMSE relation'' for the mismatched (generalized) model, 
similarly to \cite{pourkamali2022mismatchedsym,pourkamali2022mismatchednonsym} for Gaussian noise or, more closely related to this work, to \cite{barbier2022price} for rotationally invariant noise models. The derivation of the relations below for the MSE (defined in \eqref{eq:MSE-n}) is the reason for the introduction of the Wigner matrix in the generalized model, instead of considering \eqref{data} directly. Note also that its (basic) proof is independent of Conjecture \ref{Conj:LogZ}, and is fully rigorous.
\begin{lemma}[Linking the log-partition function to the MSE] Let ${f}_{\epsilon,n}^{(\alpha)}(\lambda,\lambda_*):=\EE\ln Z_n(\bm{Y}_{\epsilon})/n$. For model \eqref{data} we have
\begin{equation*}
    \text{MSE}_n=1-\frac{\alpha}{\lambda}\frac{\partial{f}_{\epsilon,n}^{(\alpha)}(\lambda,\lambda_*)}{\partial\epsilon}\bigg|_{\epsilon=0}-2\alpha\sqrt{\frac{\lambda_*}{\lambda}}\frac{\partial{f}_{0,n}^{(\alpha)}(\lambda,\lambda_*)}{\partial\lambda_*}, 
\end{equation*}
where $\text{MSE}_n:=\text{MSE}_n(M_{\text{mis}}(\bm{Y}))$.
\label{lemma:I-MMSE}
\end{lemma}

\begin{proof}
First, like \cite{barbier2022price}, we define the objects
\begin{equation*}
    M_{n}:=\frac{\langle\bm{u}^*,\bm{u}\rangle\langle\bm{v}^*,\bm{v}\rangle}{mn},\;\; Q_{n}:=\frac{\langle\bm{u}^{(1)},\bm{u}^{(2)}\rangle\langle\bm{v}^{(1)},\bm{v}^{(2)}\rangle}{mn},
\end{equation*}
where $(\bu,\bv)=(\bu^{(1)},\bv^{(1)})$ and $(\bu^{(2)},\bv^{(2)})$ are two conditionally independent samples from the mismatched posterior $P_{\text{mis}}(\,\cdot \mid\bm{Y}_{\epsilon})$ defined as \eqref{posterior} but with $\bY_\epsilon$ instead of $\bY$. We have
\begin{equation*}
    \text{MSE}_n=\frac{1}{2}\big(1-2\EE\langle M_n\rangle_0+\EE\langle Q_n\rangle_0\big).
\end{equation*}
Here $\langle\,\cdot\,\rangle_{\epsilon}$ stands for the expectation under the joint law $P_{\text{mis}}(\cdot \mid\bm{Y}_{\epsilon})^{\otimes \infty}$ of all posterior samples; $\langle\,\cdot\,\rangle_{0}$ is thus the expectation under $P_{\text{mis}}(\cdot \mid\bm{Y})^{\otimes \infty}$. We straightforwardly have that
\begin{IEEEeqnarray}{rl}
    &\frac{\partial}{\partial\sqrt{\lambda_*}}f_{0,n}^{(\alpha)}(\lambda,\lambda_*)=\frac{\sqrt{\lambda}}{\alpha}\mathbb{E}\langle M_n\rangle_0, \nn
    &\frac{\partial}{\partial\sqrt{\epsilon}}f_{\epsilon,n}^{(\alpha)}(\lambda,\lambda_*) = \frac{1}{n}\sqrt{\frac{\lambda}{\alpha}}\mathbb{E}\langle \bm{u}^\mathsf{T}\bm{W}\bm{v}\rangle_{\epsilon} = \frac{\lambda\sqrt{\epsilon}}{\alpha}(1-\EE\langle Q_n\rangle_{\epsilon}), \nonumber
\end{IEEEeqnarray}
where the last equality uses simple Gaussian integration by parts. Combining the above equations gives Lemma \ref{lemma:I-MMSE}.
\end{proof}

Now, by combining Conjecture \ref{Conj:LogZ} and Lemma \ref{lemma:I-MMSE}, we can get the explicit expression for the MSE of the mismatched Bayes estimator, which would turn into a theorem if a rigorous proof of Conjecture \ref{Conj:LogZ} were provided. To do so, we will take the derivative of $\int\hat{\rho}_{\epsilon}(\d t)\ln(x-t)$ in ${g_{\lambda,\epsilon}^{(\alpha)}(x)}$ with respect to $\epsilon$, using techniques developed for analyzing the Dyson Bessel process \cite{guionnet2021large,potters2020first}. The full derivation is in Appendix \ref{sec:appendix-mse}.

\begin{conjecture}[Performance of mismatched Bayes estimator] \label{theorem-mismatch}
Consider the model \eqref{data}. Then, we have
\begin{equation}
    \!\lim_{n\to\infty}\text{MSE}_n(M_{\text{mis}}(\bY))=\frac{1}{2}\big(1-2M(\lambda,\lambda_*)+Q(\lambda,\lambda_*)\big), \label{eq:MSE-Bayes}
\end{equation}
where $M,Q$ are given by \eqref{eq:MQ-inf} (in which $\mu$ is the asymptotic noise singular density and $\mathds{1}(\cdot)$ is the indicator function).
\end{conjecture}

\begin{figure*}[!t]
\normalsize
\begin{equation}
\begin{IEEEeqnarraybox}{rl}
    M(\lambda, \lambda_*) = &\; \alpha\sqrt{\tfrac{1}{\lambda\lambda_*}}\tfrac{C_{\mu}^{(\alpha)}(\tfrac{1}{\lambda_*})-T^{(\alpha),-1}(\tfrac{\lambda\lambda_*}{\alpha}T^{(\alpha)}(C_{\mu}^{(\alpha)}(\tfrac{1}{\lambda_*})))}{T^{(\alpha)}(C_{\mu}^{(\alpha)}(\tfrac{1}{\lambda_*}))} \\
    &\qquad\cdot\>\left[\tfrac{1}{\lambda_*}C_{\mu}^{(\alpha)\prime}(\tfrac{1}{\lambda_*})(2\alpha C_{\mu}^{(\alpha)}(\tfrac{1}{\lambda_*})+\alpha+1) - T^{(\alpha)}(C_{\mu}^{(\alpha)}(\tfrac{1}{\lambda_*}))\right]\cdot \mathds{1}\big(\bar{h}\lambda_*\geq 1\cap\lambda\lambda_*>\alpha\big), \\
    Q(\lambda,\lambda_*) = &\; \Big\{1 - \tfrac{\alpha}{\lambda\lambda_*}\Big[ 1-\tfrac{C_{\mu}^{(\alpha)}(\tfrac{1}{\lambda_*})-T^{(\alpha),-1}(\tfrac{\lambda\lambda_*}{\alpha}T^{(\alpha)}(C_{\mu}^{(\alpha)}(\tfrac{1}{\lambda_*})))}{T^{(\alpha)}(C_{\mu}^{(\alpha)}(\tfrac{1}{\lambda_*}))}(2\alpha C_{\mu}^{(\alpha)}(\tfrac{1}{\lambda_*})+\alpha+1) \Big]\Big\}\cdot \mathds{1}\big(\bar{h}\lambda_*\geq 1\cap\lambda\lambda_*>\alpha\big) \\
    &\qquad+\>\Big\{1 - \tfrac{\alpha}{\lambda}\Big[ \bar{h}-\tfrac{T^{(\alpha),-1}(\bar{\gamma}^2\bar{h})-T^{(\alpha),-1}(\tfrac{\lambda\bar{\gamma}^2}{\alpha})}{\bar{\gamma}^2}(2\alpha T^{(\alpha),-1}(\bar{\gamma}^2\bar{h})+\alpha+1) \Big]\Big\}\cdot \mathds{1}\big(\bar{h}\lambda_*<1\cap\lambda>\alpha\bar{h}\big). 
\end{IEEEeqnarraybox}
\label{eq:MQ-inf}
\vspace{-5pt}
\end{equation}
\hrulefill
\vspace{-4pt}
\end{figure*}

As a consequence of Conjecture \ref{theorem-mismatch}, similarly to \cite{barbier2022price}, we have the following result for the overlap defined in \eqref{eq:overlap-mis}:
\begin{equation}
    \text{Overlap}_{\text{mis}}=\frac{M(\lambda,\lambda_*)}{Q(\lambda,\lambda_*)^{1/2}}, \label{eq:overlap-Bayes}
\end{equation}
which holds whenever the denominator is non-zero. When comparing the overlaps of different methods in numerical experiments, we will use it for the mismatched Bayes estimator.

\section{AMP and spectral estimators}

In this section, we study the AMP and spectral estimators. In particular, we obtain a deterministic characterization of the performance of AMP in the high-dimensional limit, called SE.

\subsection{AMP estimator}

We start the AMP iterations from an initialization $\bm{u}^1\in\mathbb{R}^n$ independent of $\bm{Z}$, with a positive correlation with $\bm{u}^*$:
\begin{equation}
    (\bm{u}^*,\bm{u}^1)\overset{W_2}{\longrightarrow}(U_*,U_1),\; \mathbb{E}[U_*U_1]:=\epsilon>0,\; \mathbb{E}[U_1^2]=1,
    \label{eq:init-of-gauss-amp}
\end{equation}
where $(\bm{u}^*,\bm{u}^1)\overset{W_2}{\longrightarrow}(U_*,U_1)$ means that the joint empirical distribution of $(\bm{u}^*,\bm{u}^1)$ converges to the one of $(U_*,U_1)$ in Wasserstein-2 ($W_2$) distance. We note that this initialization is impractical, but one can design a practical one from the left singular vector of the data $\bm{Y}$ \cite{montanari2021estimation,mondelli2021pca,zhong2021approximate}. Then, the AMP iteration is given by
\begin{equation}
\begin{IEEEeqnarraybox}{rClrCl}
    \bm{g}^t &=& \bm{Y}^{\mathsf{T}}\bm{u}^t-\alpha \cdot \beta_t\bm{v}^{t-1}, &\bm{v}^t &=& v_t(\bm{g}^t), \\
    \bm{f}^t &=& \bm{Y}\bm{v}^t-\alpha_{t}\bm{u}^t, &\bm{u}^{t+1} &=& u_{t+1}(\bm{f}^t),
\end{IEEEeqnarraybox}
\label{eq:gauss-amp}
\end{equation}
where we assume that $\bm{v}^0=\bm{0}$. Here, the non-linear functions $v_t, u_{t+1}:\mathbb{R}\mapsto\mathbb{R}$ are applied component-wise; $\beta_1=0$ and for $t\geq 2$, $\beta_{t}=\langle u^\prime_{t}(\bm{f}_{t-1})\rangle$; for $t\geq 1$, $\alpha_{t}=\langle v^\prime_{t}(\bm{g}_{t})\rangle$, where $u^\prime_{t}$ and $v^\prime_{t}$ denote the derivatives. Then, the AMP estimator of $(\bm{u}^*,\bm{v}^{*})$ is $(\bm{u}^t,\bm{v}^{t})$, and the one of the spike $\bm{u}^*\bm{v}^{*\mathsf{T}}$ is $M^t_{\text{AMP}}=\bm{u}^t\bm{v}^{t\mathsf{T}}$. We refer to this algorithm as \emph{Gaussian AMP} as in the symmetric case \cite{barbier2022price}, since this is the AMP normally implemented for Gaussian noise.

We now provide the SE of this Gaussian AMP. We initialize $\bar{\nu}_1=\theta\mathbb{E}[U_*U_1]=\theta\epsilon$ and $(\bar{\bm{\Delta}})_{1,1}=\mathbb{E}[U_1^2]=1$, where $\theta := \sqrt{\lambda_*\alpha}$. Then, we define the following SE for $t\geq 1$:
\begin{IEEEeqnarray}{rl}
    &(Z_1,\ldots,Z_t)\sim \mathcal{N}(0,\bar{\bm{\Omega}}_t), \nonumber\\
    &G_t=Z_t+\bar{\nu}_tV_*-\alpha\cdot\bar{\beta}_{t}V_{t-1}+\sum_{i=1}^{t-1}(\bar{\bm{B}}_t)_{t,i}V_i,\;\; V_t=v_t(G_t), \nonumber\\
    &(Y_1,\ldots,Y_t)\sim \mathcal{N}(0,\bar{\bm{\Sigma}}_t),  \label{eq:se-of-gauss-amp}\\
    &F_t=Y_t+\bar{\mu}_tU_*-\bar{\alpha}_tU_{t}+\sum_{i=1}^{t}(\bar{\bm{A}}_t)_{t,i}U_i,\;\; U_{t+1}=u_{t+1}(F_t), \nonumber
\end{IEEEeqnarray}
where we define $V_0=0$, $\bar{\beta}_1=0$ and for $t\geq 2$, $\bar{\beta}_{t}=\mathbb{E}[u^\prime_{t}(F_{t-1})]$; for $t\geq 1$, $\bar{\alpha}_t=\mathbb{E}[v^\prime_{t}(G_{t})]$. The matrices $\bar{\bm{\Delta}}_t$, $\bar{\bm{\Gamma}}_t$, $\bar{\bm{\Phi}}_t$, $\bar{\bm{\Psi}}_t$ and the vectors $\bar{\bm{\nu}}_t$, $\bar{\bm{\mu}}_t$ are defined as follows:
\begin{equation}
\begin{IEEEeqnarraybox}{rl}
    &(\bar{\bm{\mu}}_t)_i = \frac{\theta}{\alpha}\mathbb{E}[V_*V_i],\;\; (\bar{\bm{\nu}}_t)_i = \theta\mathbb{E}[U_*U_i],\;\; 1\leq i\leq t, \\
    &(\bar{\bm{\Delta}}_t)_{ij} = \mathbb{E}[U_iU_j],\;\; (\bar{\bm{\Gamma}}_t)_{ij} = \mathbb{E}[V_iV_j],\;\; 1\leq i,j\leq t, \\
    &(\bar{\bm{\Phi}}_t)_{ij} = \mathbb{E}[\partial_jU_i],\;\; 1\leq i<j\leq t, \\
    &(\bar{\bm{\Psi}}_t)_{ij} = \mathbb{E}[\partial_jV_i],\;\; 1\leq i\leq j\leq t.
\end{IEEEeqnarraybox}
\end{equation}
We emphasize that $\partial_jU_i = \partial_{y_j}U_i$ and $\partial_jV_i = \partial_{z_j}V_i$. Then, we define the matrices $\bar{\bm{\Omega}}_t$, $\bar{\bm{B}}_t$, $\bar{\bm{\Sigma}}_t$, $\bar{\bm{A}}_t$ as follows
\begin{equation}
\begin{IEEEeqnarraybox}{rl}
    \bar{\bm{\Omega}}_t=\alpha\sum_{j=0}^{2t-2}\bar{\kappa}_{2(j+1)}\bar{\bm{\Theta}}_t^{(j)},&\;\; \bar{\bm{B}}_t=\alpha\sum_{j=0}^{t-1}\bar{\kappa}_{2(j+1)}\bar{\bm{\Phi}}_t(\bar{\bm{\Psi}}_t\bar{\bm{\Phi}}_t)^j, \\
    \bar{\bm{\Sigma}}_t=\sum_{j=0}^{2t-1}\bar{\kappa}_{2(j+1)}\bar{\bm{\Xi}}_t^{(j)},&\;\; \bar{\bm{A}}_t=\sum_{j=0}^{t}\bar{\kappa}_{2(j+1)}\bar{\bm{\Psi}}_t(\bar{\bm{\Phi}}_t\bar{\bm{\Psi}}_t)^j,
\end{IEEEeqnarraybox}
\label{eq:corr-ABSigmaOmega}
\end{equation}
where
\begin{equation*}
\begin{IEEEeqnarraybox}{rCl}
    \bar{\bm{\Theta}}_t^{(j)} &=& {\textstyle\sum_{i=0}^j}(\bar{\bm{\Phi}}_t\bar{\bm{\Psi}}_t)^i\bar{\bm{\Delta}}_t(\bar{\bm{\Psi}}_t^\mathsf{T}\bar{\bm{\Phi}}_t^\mathsf{T})^{j-i} \\
    && +\> {\textstyle\sum_{i=0}^{j-1}}(\bar{\bm{\Phi}}_t\bar{\bm{\Psi}}_t)^i\bar{\bm{\Phi}}_t\bar{\bm{\Gamma}}_t\bar{\bm{\Phi}}_t^\mathsf{T}(\bar{\bm{\Psi}}_t^\mathsf{T}\bar{\bm{\Phi}}_t^\mathsf{T})^{j-1-i} ,\\
    \bar{\bm{\Xi}}_t^{(j)} &=& {\textstyle\sum_{i=0}^j}(\bar{\bm{\Psi}}_t\bar{\bm{\Phi}}_t)^i\bar{\bm{\Gamma}}_t(\bar{\bm{\Phi}}_t^\mathsf{T}\bar{\bm{\Psi}}_t^\mathsf{T})^{j-i} \\
    && +\> {\textstyle\sum_{i=0}^{j-1}}(\bar{\bm{\Psi}}_t\bar{\bm{\Phi}}_t)^i\bar{\bm{\Psi}}_t\bar{\bm{\Delta}}_t\bar{\bm{\Psi}}_t^\mathsf{T}(\bar{\bm{\Phi}}_t^\mathsf{T}\bar{\bm{\Psi}}_t^\mathsf{T})^{j-1-i} ,
\end{IEEEeqnarraybox}
\end{equation*}
and $\{\bar{\kappa}_{2j}\}_{j\geq{1}}$ denotes the sequence of rectangular free cumulants associated to the limit of the singular law of $\bm{Z}$. We assume that the non-linear functions $v_{t}$ and $u_{t+1}$ are pseudo-Lipschitz of order $2$, where a function $\psi$ is pseudo-Lipschitz of order $2$ if there is a constant $C>0$ such that $\|\psi(\bm{x})-\psi(\bm{y})\|\leq C(1+\|\bm{x}\|+\|\bm{y}\|)(\|\bm{x}-\bm{y}\|)$. We also assume that the partial derivatives
\vspace{-5pt}
\begin{equation*}
\begin{IEEEeqnarraybox}{rl}
    &\partial_{z_j}v_t(Z_t+\bar{\nu}_tV_*-\alpha\cdot\bar{\beta}_{t}V_{t-1}+\sum_{i=1}^{t-1}(\bar{\bm{B}}_t)_{t,i}V_i), \\
    &\partial_{y_j}u_{t+1}(Y_t+\bar{\mu}_tU_*-\bar{\alpha}_tU_{t}+\sum_{i=1}^{t}(\bar{\bm{A}}_t)_{t,i}U_i)
\end{IEEEeqnarraybox}
\end{equation*}
are continuous on a set of probability $1$, under the laws of $(Z_1,\ldots,Z_t)$, $(V_1,\ldots,V_t)$, $(Y_1,\ldots,Y_t)$ and $(U_1,\ldots,U_t)$ given in \eqref{eq:se-of-gauss-amp}. Then, we have the following rigorous SE for Gaussian AMP. Note that this theorem does \emph{not} follow from the standard SE analysis for Gaussian AMP when the noise is actually Gaussian \cite{fletcher2018iterative}, due to the mismatch.
\begin{theorem}[State evolution of Gaussian AMP] \label{theorem:se-of-gauss-amp}
Consider model \eqref{data}, the Gaussian AMP and its SE above. Let $\phi:\mathbb{R}^{2t+1}\mapsto\mathbb{R}$ and $\psi:\mathbb{R}^{2t+2}\mapsto\mathbb{R}$ be any pseudo-Lipschitz functions of order $2$. For each $t\geq 1$, we almost surely have 
\vspace{-5pt}
\begin{IEEEeqnarray}{rCl}
    &\lim_{m\to\infty}\frac{1}{m}\sum_{i=1}^m&\phi\big((\bm{g}^1)_i,\ldots,(\bm{g}^t)_i,(\bm{v}^1)_i,\ldots,(\bm{v}^t)_i,(\bm{v}^*)_i\big) \nonumber \\
    &&=\mathbb{E}\phi\big(G_1,\ldots,G_t,V_1,\ldots,V_t,V_*\big), \label{eq:se-of-gauus-amp-1} \\
    &\lim_{n\to\infty}\frac{1}{n}\sum_{i=1}^n&\psi\big((\bm{f}^1)_i,\ldots,(\bm{f}^t)_i,(\bm{u}^1)_i,\ldots,(\bm{u}^{t+1})_i,(\bm{u}^*)_i\big) \nonumber \\
    &&=\mathbb{E}\psi\big(F_1,\ldots,F_t,U_1,\ldots,U_{t+1},U_*\big) .\label{eq:se-of-gauus-amp-2}
\end{IEEEeqnarray}
\end{theorem}
\begin{proof}[Proof idea]
We carefully choose the non-linear functions $\tilde{v}_{t}$ and $\tilde{u}_{t+1}$ to design an auxiliary AMP algorithm such that \emph{(i)} there is an SE characterization, and \emph{(ii)} its iterations are close to the Gaussian AMP. The full proof is in Appendix~\ref{sec:appendix-seproof}.
\end{proof}

\subsection{Spectral estimators} 

Spectral estimators are of the form $J\bm{u}_1\bm{v}_1^\mathsf{T}$ where $\bm{u}_1$ and $\bm{v}_1$ are the singular vectors associated to the largest singular value $\sigma_1(\bm{Y})$, with norm $\sqrt{n}$ and $\sqrt{m}$ respectively. 
We consider the following two spectral estimators:
\begin{itemize}
    \item The \emph{optimal spectral estimator (OptSpec)} is $M_{\text{OS}}=J_{\text{OS}}\bm{u}_1\bm{v}_1^\mathsf{T}$, where $J_{\text{OS}}=J(\mu,\lambda_*)$, which depends on the limit of the singular law of $\bm{Z}$ and the SNR $\lambda_*$, is the optimal scaling in the sense that it is not mismatched and thus minimizes the MSE over all spectral estimators.
    \item The \emph{Gaussian mismatched spectral estimator (GauSpec)} is $M_{\text{GS}}=J_{\text{GS}}\bm{u}_1\bm{v}_1^\mathsf{T}$, where $J_{\text{GS}}=J(\mu_{\text{G}},\lambda)$ is the optimal scaling if the noise would be a Gaussian matrix and the true SNR $\lambda$. It thus depends on the singular law $\mu_{\text{G}}$ of a Gaussian matrix and $\lambda$ instead of $(\mu,\lambda_*)$. 
\end{itemize}
The explicit expression of the scaling constant $J(\mu,\lambda_*)$ is given in Theorem \ref{theo:J(mu,lam)} of Appendix \ref{sec:appendix-lowrank}. In both cases, the estimators of $\bm{u}^*$ and $\bm{v}^*$ are $\bm{u}_1$ and $\bm{v}_1$, respectively. From \cite{benaych2012singular}, we know that the overlap $\langle\bm{u}_1,\bm{u}^*\rangle\langle\bm{v}_1,\bm{v}^*\rangle/(mn)$ converges to $J_{\text{OS}}$. This allows us to deduce that $\text{MSE}(M_{\text{OS}})=\frac{1}{2}(1-J_{\text{OS}}^2)$ and $\text{MSE}(M_{\text{GS}})=\frac{1}{2}(1+J_{\text{GS}}^2-2J_{\text{GS}}J_{\text{OS}})$. From $\text{MSE}(M_{\text{GS}})=\text{MSE}(M_{\text{OS}})+\frac{1}{2}(J_{\text{OS}}-J_{\text{GS}})^2$, we have that $\text{MSE}(M_{\text{GS}})\geq\text{MSE}(M_{\text{OS}})$ (the Gaussian estimator is mismatched but the optimal one is not).

Note that when $\lambda=\lambda_*$, there does not appear to be any special relationship between $\text{MSE}(M_{\text{mis}})$ and $\text{MSE}(M_{\text{GS}})$. Indeed, if $\bar{h}>1/\sqrt{\alpha}$, then $\text{MSE}(M_{\text{mis}})$ and $\text{MSE}(M_{\text{GS}})$ differ (although they are numerically close in our experiments). Interestingly, this property differs from the square case \cite{barbier2022price}, where $\text{MSE}(M_{\text{mis}})=\text{MSE}(M_{\text{GS}})$ when $\lambda=\lambda_*$ and $\bar{h}>1$.

\section{Numerical experiments}

In all numerical experiments, we make sure that the density $\mu$ has unit variance. We separate the effect of the mismatch in the noise statistics and in the SNR. More specifically, we consider mismatched noise statistics but matched SNR $\lambda=\lambda_*$ in the first example; in the second example, the noise is matched while the SNR is not.

In the first example, the asymptotic empirical singular density of the noise is the rectangular analogue of the symmetrized Poisson distribution with parameter $c$, which can be seen as the weak limit of $((1-c/n)\delta_0+(c/(2n))(\delta_{-1}+\delta_1))^{*n}$ ($f^{*n}$ means the convolution of a function $f$ with itself $n$ times) and whose rectangular R-transform is $\frac{cz}{1-z}$, see \cite[Section 4.3]{benaych2007infinitely}. We set $c=1$ to enforce unit variance. Furthermore, the rectangular analogue of the symmetrized Poisson distribution with parameter $c=1$ is the limit of the singular law of a random matrix $M(d,d^\prime)=\sum_{k=1}^{d}u_{d}(k)v_{d^\prime}(k)^\mathsf{T}$ where $d\to\infty$, $d/d^\prime\to\alpha$ and $u_{d}(k),v_{d^\prime}(k)\;(k\geq 1)$ are independent uniform random vectors on the unit spheres of $\mathbb{R}^d, \mathbb{R}^{d^\prime}$, see \cite[Proposition 6.1]{benaych2007infinitely}. This characterization allows us to simulate the rectangular analogue of the symmetrized Poisson distribution, and calculate its theoretical Bayes MSE. We compare the Bayes estimator to the \emph{correct AMP} implemented in \cite{fan2022approximate} (based on previous work \cite{opper2001adaptive}), and we call it ``correct'' as it is conjectured to be Bayes-optimal for low-rank matrix estimation with structured noise in the recent work \cite{barbier2022bayes}. However, we emphasize that the AMP of \cite{fan2022approximate} is provably not optimal anymore, when the signals $(\bm{u}^*,\bm{v}^*)$ are not uniformly distributed on the sphere or Gaussian. In this case, the conjectured Bayes-optimal AMP was recently introduced in \cite{barbier2022bayes} for the symmetric case.

We observe the following phenomena, some of which are in correspondence to the symmetric case \cite{barbier2022price}. First, the Gaussian AMP does not perform as well as the mismatched Bayes estimator, especially when $\lambda_*$ is small. Second, when $\lambda_*$ is large enough, the overlaps of all estimators match and saturate the one of the optimal estimators (correct AMP and OptSpec). Thus, the gap in the MSE comes from the incorrect estimation of the signal norm, not from the overlap (i.e., its ``direction''). Third, when there is no SNR mismatch and $\bar{h}>{1}/\sqrt{\alpha}$, we observe that $\text{MSE}(M_{\text{mis}})$ almost matches $\text{MSE}(M_{\text{GS}})$. However, we emphasize that these two values are not exactly the same. Fourth, all estimators are outperformed by the two optimal estimators. Lastly, Bayes and Gaussian spectral MSE curves are non-decreasing with the true SNR $\lambda_*$.

\begin{figure}[!t]
    \centering
    \subfloat{\includegraphics[width=0.45\textwidth, trim={90 210 90 190}]{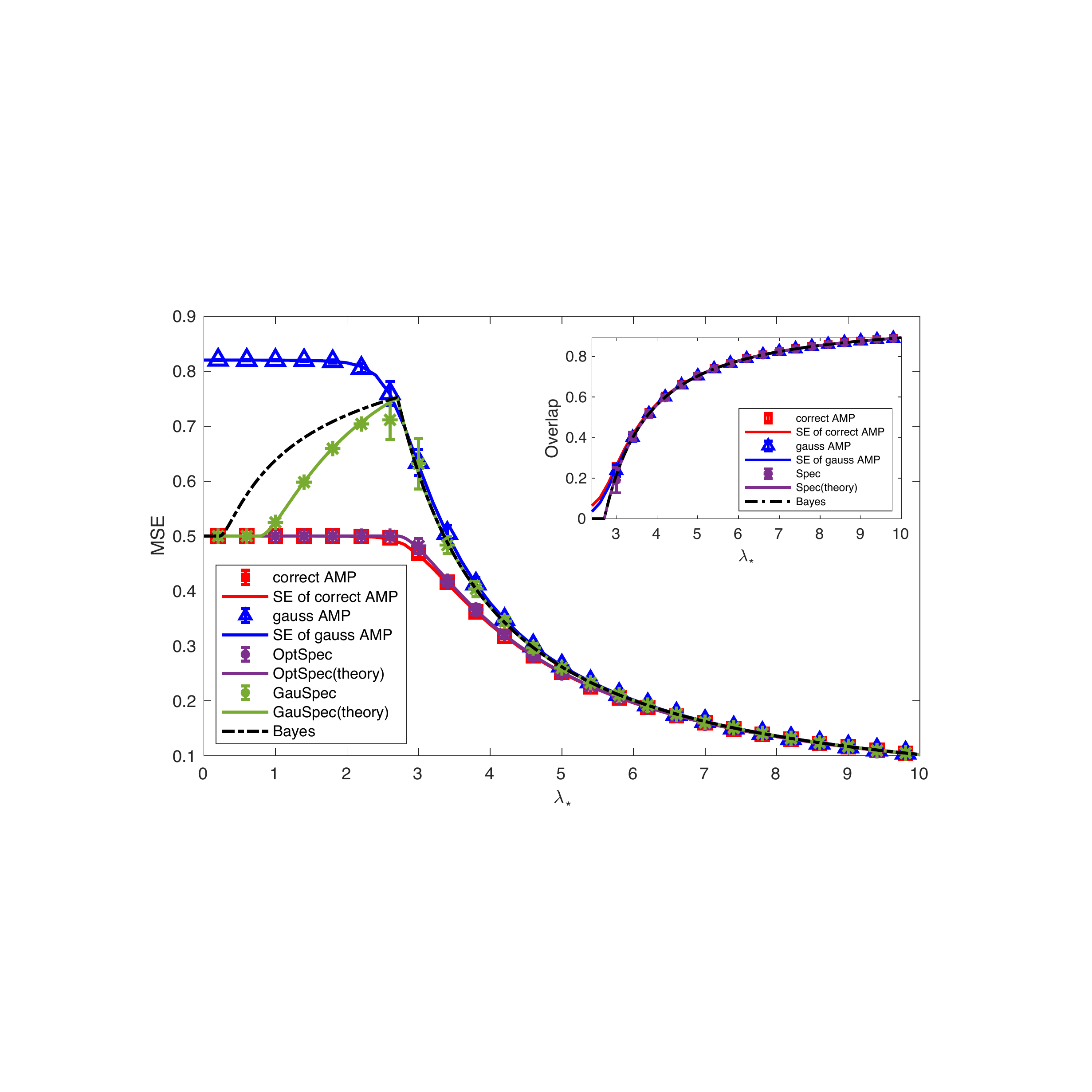}}
    \hfill
    \subfloat{\includegraphics[width=0.45\textwidth, trim={90 190 90 195}]{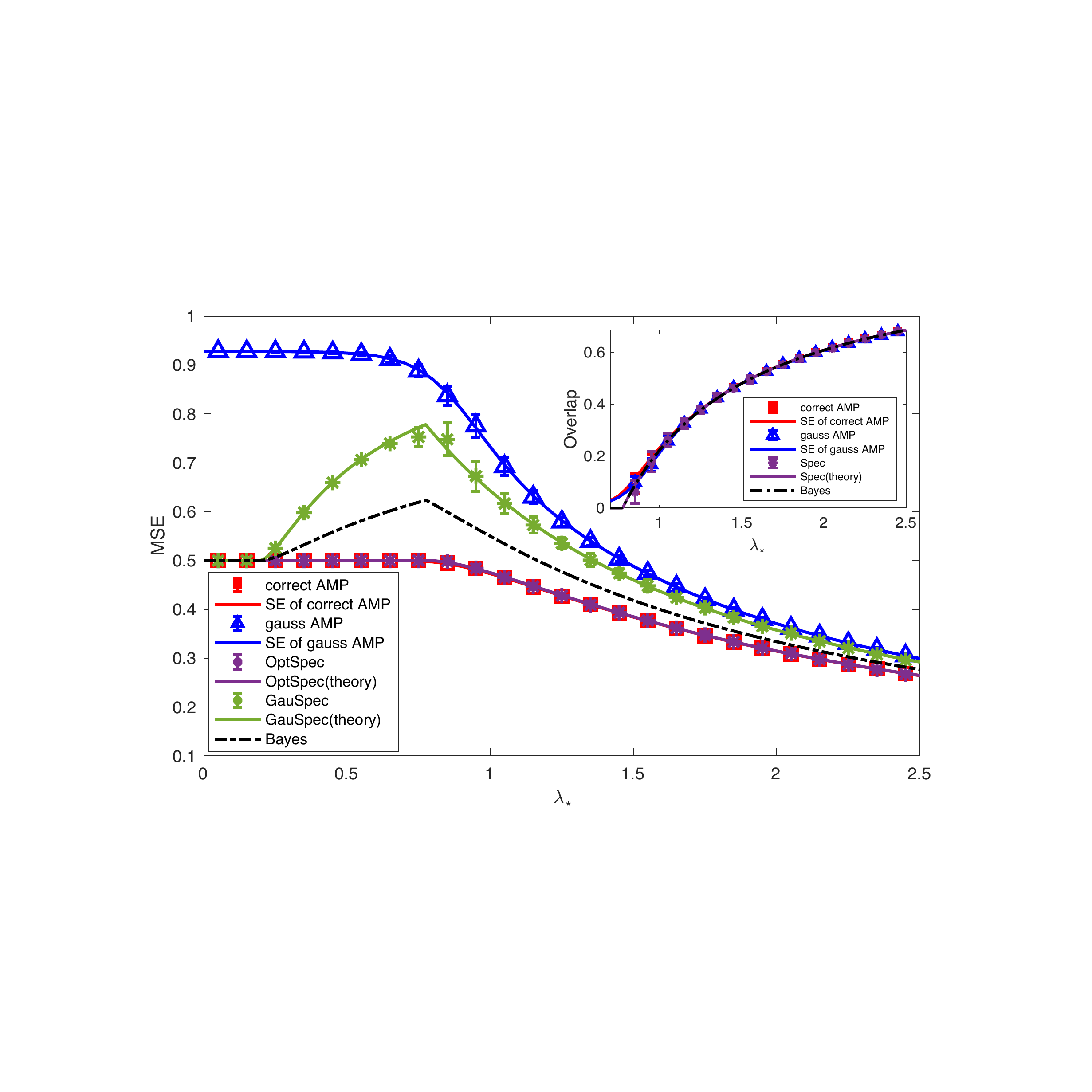}}
    \hfill
    \caption{MSE and overlap as a function of the true SNR $\lambda_*$ in two mismatched settings. The first figure corresponds to a setting with matched SNR and mismatched noise, with a noise drawn according to the rectangular analogue of the symmetrized Poisson distribution \cite{benaych2007infinitely}. For the second, we do the opposite: the noise is Gaussian (and thus matched) while the SNR is mismatched as $\lambda=4\lambda_*$. We used $m=20000$ for AMP, and $m=10000$ for spectral estimators. We set $\alpha=n/m=0.6$. The points are averaged over $100$ trials for AMP and over $20$ for spectral estimators.}
    \label{fig:MSE}
\end{figure}

\section{Perspectives}

Together with \cite{barbier2022price}, the present paper uncovers a surprising phenomenology in various algorithmic behaviors, due to the presence of mismatch in rank-$1$ matrix estimation. It would be interesting to clarify the generality of these observations; in particular, whether they extend or how they differ in multiview versions of matrix estimation \cite{9173970,zhong2021approximate}, when the noise is inhomogeneous \cite{guionnet2022low}, or in tensor estimation \cite{de2022random,seddik2021random}. Another direction is to compare the cost of the mismatch quantified in the present work to the Bayes-optimal performance, which requires extending \cite{barbier2022bayes} to non-symmetric settings.




\input{ref.bbl}


\appendices

\section{Low-rank perturbations of rotationally invariant matrices}
\label{sec:appendix-lowrank}

In this appendix, we recall some known results concerning low-rank perturbations of rotationally invariant matrices \cite{benaych2012singular}. The first result is about the largest singular value of $\bm{Y}$ in the presence of a rank-one perturbation.
\begin{theorem}[Theorem 2.8 of \cite{benaych2012singular}] \label{theorem:BBP-transition}
Consider the spike model \eqref{data}. Then, as $n\to\infty$, the largest singular value $\bar{\nu}$ of $\bm{Y}$ converges almost surely to
\begin{equation}
    \bar{\nu}=D^{-1}_{\mu}\big(\frac{1}{\lambda_*}\big)\mathds{1}\big(\bar{h}\lambda_*\geq 1\big) + \bar{\gamma}\mathds{1}\big(\bar{h}\lambda_*<1\big). \label{eq:BBP}
\end{equation}
\end{theorem}

The second result is about the overlap of the true signal and the singular vector associated to the largest singular value.
\begin{theorem}[Theorem 2.9 of \cite{benaych2012singular}] \label{theo:J(mu,lam)}
Consider the spike model \eqref{data}. Then, as $n\to\infty$, the singular vectors $\bm{u}_1$ and $\bm{v}_1$ corresponding to the largest singular value $\bar{\nu}$ of $\bm{Y}$ satisfy \eqref{eq:singular-vector} almost surely.
\end{theorem}

\begin{figure*}[!b]
\normalsize
\vspace{4pt}
\hrulefill
\begin{equation}
    J(\mu,\lambda_*):=\lim_{n\to\infty}\frac{|\langle\bm{u}_1,\bm{u}^*\rangle\langle\bm{v}_1,\bm{v}^*\rangle|}{mn}=\frac{|T^{(\alpha)}(C_{\mu}^{(\alpha)}(\frac{1}{\lambda_*}))-\frac{1}{\lambda_*}C_{\mu}^{(\alpha)\prime}(\frac{1}{\lambda_*})(2\alpha C_{\mu}^{(\alpha)}(\frac{1}{\lambda_*})+\alpha+1)|}{\sqrt{T^{(\alpha)}(C_{\mu}^{(\alpha)}(\frac{1}{\lambda_*}))}}\cdot\mathds{1}(\bar{h}\lambda_*\geq 1) \label{eq:singular-vector}
\end{equation}
\end{figure*}

\section{Computation of the log-partition function} \label{sec:appendix-logfunc}

In this section, we will specifically analyze the ``high and low temperatures'' $\theta$ regimes of the rectangular spherical integral and locate the transition point between the two.

\subsection{High temperature regime} \label{seq:equivalentSphInt}

Once we have recognized that the partition function is a rectangular spherical integral, we can use the main result from \cite{benaych2011rectangular} which rigorously shows that for a sufficiently small $\theta$,
\begin{equation}
    \lim_{n\to \infty}\frac{1}{n}\ln I_n(\theta,\bm{Y}) = \int_{0}^{\theta}\frac{C_{\mu}^{(\alpha)}(t^2)}{t}\d t .\label{Flor-int}
\end{equation}
The function
$C^{(\alpha)}_{\mu}$ is the so-called \emph{rectangular R-transform with ratio $\alpha$} \cite{benaych2011rectangular}. Therefore, in this section we show that the formulations \eqref{Flor-int} and \eqref{Kaba-int} for the rectangular spherical integral are equivalent in a certain regime of ``temperatures''. This regime corresponds to values of $\theta$ such that $\nabla \phi=\boldsymbol{0}$ does possess a solution, and we call it the high temperature regime of the rectangular spherical integral. 

Let us make explicit the equations verified by the stationary point. By using the R-transform of $\hat{\rho}$ (see, e.g., \cite{potters2020first}), we have
\begin{equation}
    \frac{\partial \phi}{\partial z_2} = 0\;\;\Rightarrow\;\; z_2=\frac{\theta^2}{z_1}R_{\hat{\rho}}\bigg(\frac{\theta^2}{z_1}\bigg)+1. \label{18}
\end{equation}
A useful identity it verifies (which can be obtained by derivation on both sides with respect to $\theta$ and noticing that it is verified at $\theta=0$) is \cite{guionnet2005fourier}
\begin{equation*}
    \theta R_\rho(\theta) - \int\rho(\dif t)\ln(\theta R_\rho(\theta)-\theta t+1) = \int_{0}^{\theta}R_\rho(x)\dif x.
\end{equation*}
Hence, 
\begin{equation}
    \phi(z_1, z_2(z_1)) = \int_0^{\frac{\theta^2}{z_1}}R_{\hat{\rho}}(t)\d t + \alpha z_1 - \alpha\ln z_1 + 1. \label{equ-temp1}
\end{equation}
Thus, the stationary condition for $z_1$ then satisfies
\begin{equation}
    \frac{\d}{\d z_1}\phi(z_1, z_2(z_1)) = - \frac{\theta^2}{z_1^2}R_{\hat{\rho}}\bigg(\frac{\theta^2}{z_1}\bigg) + \alpha - \frac{\alpha}{z_1} = 0. \label{equ-temp3}
\end{equation}

We want to show that, when $(z_1,z_2)$ is solution to the above pair of saddle point equations \eqref{18} and \eqref{equ-temp3}, then
\begin{equation}
    \frac{1}{2\alpha}\left[\phi(z_1, z_2(z_1)) - (1+\alpha)\right] = \int_{0}^{\theta}\frac{C_{\mu}^{(\alpha)}(t^2)}{t}\d t .\label{equ-temp2}
\end{equation}
Note that this equation is correct when $\theta=0$: taking $\theta=0$ into (\ref{equ-temp3}) yields $z_1=1$; doing the same in (\ref{equ-temp1}) gives $\phi(z_1, z_2(z_1))|_{\theta=0}=\alpha+1$, so that \eqref{equ-temp2} becomes $0=0$. Next, let us take the derivative w.r.t. $\theta$ and then multiply by $\theta$ on both sides of (\ref{equ-temp2}) (because we are at an extremum $\partial_{z_1} \phi=0$). This gives
\begin{equation}
    z_1-1 =\frac{1}{\alpha}\frac{\theta^2}{z_1}R_{\hat{\rho}}\bigg(\frac{\theta^2}{z_1}\bigg) = \frac{1}{\alpha}K_{\hat{\rho}}\bigg(\frac{\theta^2}{z_1}\bigg) \overset{?}{=} C_{\mu}^{(\alpha)}(\theta^2) =: \gamma(\theta), \label{prove1-tmp1}
\end{equation}
where $K_{\hat{\rho}}(z):=zR_{\hat{\rho}}(z)$. The first equality comes directly from (\ref{equ-temp3}) and the question mark is what we need to prove. Denoting simply $\gamma=\gamma(\theta)$, if we define $M_{\rho}(z) = \sum_{k\geq 1}m_{k}(\rho)z^k$ where $m_{k}(\rho) := \int \rho(\dif t)t^{2}$, from Lemma 3.2 of \cite{benaych2011rectangular} we have that 
\begin{equation}
    M_{\rho}\left(\frac{\theta^2}{(\alpha\gamma+1)(\gamma+1)}\right) = \gamma.
\end{equation}
From the relationship between free cumulants and moments, we also have $K_{\hat{\rho}}(z(\alpha M_{\rho}(z)+1)) = \alpha M_{\rho}(z)$. Here we use the fact that $M_{\hat{\rho}}(z) = \alpha M_{\rho}(z)$ and the relation $K_{\hat{\rho}}\big(z(M_{\hat{\rho}}(z)+1)\big)=M_{\hat{\rho}}(z)$ from \cite{benaych2008free}. Consider $z=\theta^2/((\alpha\gamma+1)(\gamma+1))$ into the above to reach
\begin{equation}
    K_{\hat{\rho}}\bigg(\frac{\theta^2}{\gamma+1}\bigg) = \alpha M_{\rho}\left(\frac{\theta^2}{(\alpha\gamma+1)(\gamma+1)}\right) = \alpha\gamma. \label{prove1-tmp2}
\end{equation}
Finally, comparing with the left side of (\ref{prove1-tmp1}) and (\ref{prove1-tmp2}) we obtain 
\begin{equation}
  \gamma(\theta)=z_1(\theta)-1  \label{gammaz1}
\end{equation}
which shows \eqref{prove1-tmp1}. Therefore, whenever $(z_1,z_2)$ is solution to the stationary conditions \eqref{18} and \eqref{equ-temp3}, then \eqref{equ-temp2} holds, and also $C_{\mu}^{(\alpha)}(\theta^2)=z_1-1$.

To complete the argument, we now show that, when a stationary solution exists, then it is unique. We write again the two critical conditions:
\begin{IEEEeqnarray}{rl}
    \frac{\partial \phi}{\partial z_1} &= -\int\hat{\rho}(\d t)\frac{z_2}{z_1z_2-\theta^2t} + \alpha - \frac{\alpha-1}{z_1} = 0 \label{derivative-z1},\\
    \frac{\partial \phi}{\partial z_2} &= -\int\hat{\rho}(\d t)\frac{z_1}{z_1z_2-\theta^2t} + 1 = 0. \label{derivative-z2}
\end{IEEEeqnarray}
Notice that $z_1>0, z_2>0$. Combining these two equations we get that 
\begin{equation}
    z_2=z_2(z_1) = \alpha z_1 - \alpha + 1. \label{equa-z12-1}
\end{equation}
This equation together with \eqref{derivative-z2} fixes uniquely the solution $(z_1,z_2)$ to the stationary equations. Indeed, when we take (\ref{equa-z12-1}) into (\ref{derivative-z2}), we can get that
\begin{equation}
    \int\hat{\rho}(\d t)\Big(z_2(z_1)-\frac{\theta^2t}{z_1}\Big)^{-1} =1.\label{critCond2}
\end{equation}
This integral is strictly decreasing with respect to $z_1$. Thus, if it has a solution for $z_1$, then it is unique. In this case, $z_2$ is also unique because of (\ref{equa-z12-1}).

To summarize: if the system \eqref{derivative-z1}-\eqref{derivative-z2} or equivalently  \eqref{equa-z12-1}-\eqref{critCond2} has a solution (which we now know is unique), then it must be plugged in \eqref{Kaba-int} and this will lead to the equivalent expression \eqref{Flor-int} written in terms of the rectangular R-transform.

\subsection{Low temperature regime}
Let us consider the scenario in which there is no stationary solution. This means that no $(z_1,z_2)\in \mathcal{D}(\theta,\bar\nu)$ verifies the system \eqref{derivative-z1}-\eqref{derivative-z2} (or equivalently \eqref{derivative-z2}-\eqref{equa-z12-1}). We will call this scenario the ``low temperature regime'' of the rectangular spherical integral. Here, the extremum corresponds to the boundary of the domain $\mathcal{D}(\theta,\bar\nu)$. This means that one needs to choose $z_2$ given by \eqref{equa-z12-1} while $z_1$ is given by the following equation defining the boundary of $\mathcal{D}(\theta,\bar\nu)$ (recall that $\bar{\nu}$ is the largest singular value of $\bm{Y}$):
\begin{equation}
    z_1z_2 = \theta^2\bar{\nu}^2. \label{equa-z12-2}
\end{equation}

We now explain why (\ref{equa-z12-1}) and (\ref{equa-z12-2}) determine the value of (\ref{Kaba-int}) when $\theta$ is not small and therefore no stationary solution of $\nabla\phi=\boldsymbol{0}$ exists. The function $\phi$ in \eqref{func-phi12} should be understood as the limit as $n\to \infty$ of the following sequence of functions:
\begin{equation*}
    \phi_n(z_1,z_2)\!:=\!-\frac1n \sum_{k\le n}\ln(z_1z_2-\theta^2\lambda_k)+z_1\alpha+z_2-(\alpha-1)\ln z_1, 
\end{equation*}
where $(\lambda_k)_{k\le n}$ are the eigenvalues of the matrix $\bY\bY^\mathsf{T}$,  which we consider ordered from largest $\lambda_1$ to smallest $\lambda_n$. This function $\phi_n$ represents an action that needs to be extremized to get the expression of the rectangular spherical integral by steepest descent. Thus, we need to look for stationary solutions of $\nabla\phi_n=\boldsymbol{0}$ for large but finite $n$ (not of $\nabla\phi=\boldsymbol{0}$ directly). The unique solution of $\nabla\phi_n=\boldsymbol{0}$, which really is what needs to be plugged in the action $\phi_n$ when evaluating the spherical integral for large $n$, asymptotically matches the solution of $\nabla\phi=\boldsymbol{0}$ only in the high temperature regime.

The stationary conditions $\nabla\phi_n=\boldsymbol{0}$ read like \eqref{derivative-z1}-\eqref{derivative-z2}, but with the integral over the asymptotic density $\hat \rho(t)$ replaced by an empirical expectation over the eigenvalues $(\lambda_k)_{k\le n}$. Therefore, the condition \eqref{equa-z12-1} still holds. Instead, the stationary condition \eqref{critCond2} now reads
\begin{equation*}
    \frac1n \sum_{2\le k\le n}\Big(z_2(z_1)-\frac{\theta^2\lambda_k}{z_1}\Big)^{-1}+\frac1n\Big(z_2(z_1)-\frac{\theta^2\lambda_1}{z_1}\Big)^{-1} =1. 
\end{equation*}
The above sum $n^{-1}\sum_{2\le k\le n}(\cdots)^{-1}$ (which is decreasing in $z_1$) has same limit as the integral in \eqref{critCond2} which, by hypothesis of the lack of existence of a stationary solution of $\nabla\phi=\boldsymbol{0}$, remains bounded below $1$ for $n\to  \infty$ for any value of $z_1$ in $\mathcal{D}(\theta,\bar\nu)$. Thus, the second term must be large enough to compensate and fulfill the stationary conditions above, meaning that it requires 
$$z_2(z_1)z_1- \theta^2\lambda_1 =\Theta(n^{-1}).$$
The limit of $\lambda_1(n)$ is $\bar\nu^2$, so in the large $n$ limit we obtain \eqref{equa-z12-2}. Therefore, when $\nabla\phi=\boldsymbol{0}$ has no solution, the solution $(z_1,z_2)$ of $\nabla\phi_n=\boldsymbol{0}$ for large $n$ ``sticks'' to a constant value for all $\theta$ large enough, which is the analogue of what happens in the standard spherical integral \cite{potters2020first,guionnet2005fourier}. In this ``low temperature regime'', we thus get the value of $(z_1,z_2)$ by solving (\ref{equa-z12-1}), (\ref{equa-z12-2}):
\begin{IEEEeqnarray}{rCl}
    z_1&=&z_1(\alpha,\theta,\bar \nu)=T^{(\alpha), -1}(\theta^2\bar{\nu}^2)+1, \label{equa-z1final}\\
    z_2&=&z_2(z_1) = \alpha z_1(\alpha,\theta,\bar \nu) - \alpha + 1. \label{equa-z12-1-2}
\end{IEEEeqnarray}
Plugging \eqref{equa-z12-2}, \eqref{equa-z1final}, \eqref{equa-z12-1-2} into \eqref{Kaba-int}, \eqref{func-phi12}, the log-rectangular spherical integral at low temperature reads
\begin{IEEEeqnarray}{rl}
    -\frac1{2\alpha}\int\hat{\rho}(\d t)\ln(\bar\nu^2-t)+ z_1(\alpha,\theta,\bar\nu)&-\frac{\alpha-1}{2\alpha}\ln z_1(\alpha,\theta,\bar\nu) \nn
    &-\frac{\ln|\theta|}{\alpha}-1. \label{gsoon}
\end{IEEEeqnarray}

\subsection{Finding the phase boundary}

We aim at finding the value $\bar \theta$ of the temperature $\theta$ which separates the aforementioned high and low temperature phases. From all the previous explanations, we know that the solution $(z_1,z_2)$ is continuous in $\theta$ (it clearly is continuous in the high temperature phase $z_1z_2 > \theta^2\bar{\nu}^2$ and then sticks to the boundary $z_1z_2 = \theta^2\bar{\nu}^2$ when entering in the low temperature one). Therefore, the transition point $\bar \theta$ is defined by the condition that, if $\theta < \bar \theta$, then \eqref{equa-z12-1}, \eqref{critCond2} hold and the spherical integral is given by \eqref{Flor-int} or \eqref{Kaba-int} which match; if instead $\theta > \bar \theta$, then \eqref{equa-z12-1}, \eqref{equa-z1final} hold and the variational form \eqref{Kaba-int} must necessarily be considered, and precisely at $\theta=\bar\theta$ all these conditions hold jointly and both formulations of the spherical integral match. Therefore, if we plug \eqref{equa-z1final}, \eqref{equa-z12-2} into \eqref{critCond2} the condition on $\bar \theta$ can be written as
\begin{equation*}
    \lim_{z\downarrow\bar{\nu}^2}\frac{\alpha-1+((\alpha-1)^2+4\alpha\bar{\theta}^2z)^{1/2}}{2\alpha\bar{\theta}^2}\int \hat{\rho}(\d t)\frac{1}{z-t} = 1,
\end{equation*}
where the $\lim$ is to ensure that the value on the left is meaningful. Recall $H_{\hat{\rho}}(z):=\int \hat{\rho}(\d t)(z-t)^{-1}$ and assume that $\lim_{z\downarrow\bar{\nu}^2}H_{\hat{\rho}}(z)$ is finite (a standard example of such density is the Marcenko-Pastur law). Then from above we have
\begin{equation*}
    \bar{\theta} = \sqrt{D_{\mu}(\bar{\nu}^+)}.
\end{equation*}
Here, $D_{\mu}>0$ is the D-transform of $\mu$. If $z\ge \bar \nu^2$, then $z\ge \bar \gamma^2$ so we can use this definition. We thus have obtained the temperature at which the phase transition separating the two regimes occurs, see \eqref{theta_trans}.

\subsection{Combining everything to get the log-partition function}

Now that we have expressed the log-rectangular spherical integral in both temperature regimes and found the transition point $\bar \theta$ between these, we are ready to get the log-partition function. The first transition is the behavior of the limit of the largest singular values $\bar \nu$ of the data $\bY$. This question was studied in \cite{benaych2012singular} and we recall it in Theorem \ref{theorem:BBP-transition} of Appendix \ref{sec:appendix-lowrank}. The second transition that will play a role is the one between the high and low temperature behaviors of the rectangular spherical integral dictated by \eqref{theta_trans}, which reads: 
\begin{equation}
    {\mbox{Sticking transition}}:\quad \lambda/\alpha = D_{\mu}^{(\alpha)}(\bar{\nu}^+). \label{sticking}
\end{equation}

Our conjecture for the generalized observation model, which matches the initial one \eqref{data} when setting $\epsilon=0$, is stated in Conjecture \ref{Conj:LogZ} which is obtained straightforwardly by combining all our previous results, in particular: the distinction between high \eqref{Flor-int} and low \eqref{gsoon} temperature expressions of the rectangular spherical integral, and the identification of the transition point \eqref{sticking} between these two.

\section{Computation of the mean-square error}\label{sec:appendix-mse}
Now that we obtained a formula for the log-partition function, we can derive the expression for the mean-square error (MSE) of the mismatched Bayesian statistician. In this section, we will omit the estimator symbol $M_{\text{mis}}(\bm{Y})$ because we only consider the mismatched Bayes estimator. From the details of Lemma \ref{lemma:I-MMSE}, we can define $M(\lambda,\lambda_*)$ and $Q(\lambda,\lambda_*)$ as follows
\begin{IEEEeqnarray}{rl}
    M(\lambda,\lambda_*) &:= \lim_{n\to\infty}\mathbb{E}\langle M_n\rangle_0 = 2\alpha\sqrt{\frac{\lambda_*}{\lambda}}\frac{\partial f_{0}(\lambda,\lambda_*)}{\partial\lambda_*}, \label{M-flambda} \\
    Q(\lambda,\lambda_*) &:= \lim_{n\to\infty}\mathbb{E}\langle Q_n\rangle_0 = 1-\frac{2\alpha}{\lambda}\frac{\partial f_{\epsilon}(\lambda,\lambda_*)}{\partial\epsilon}\bigg|_{\epsilon=0} .\label{Q-flambda}
\end{IEEEeqnarray}
The limit of MSE when $n\to\infty$ is then
\begin{equation}
    \lim_{n\to\infty}\text{MSE}_n=\big(1-2M(\lambda,\lambda_*)+Q(\lambda,\lambda_*)\big)/2. 
\end{equation}
Therefore, we need to compute the derivatives \eqref{M-flambda}, \eqref{Q-flambda} in the various temperature (i.e., SNR) regimes dictated by Conjecture \ref{Conj:LogZ}, so that we complete the proof of the Conjecture \ref{theorem-mismatch}.

\subsection{Low temperature regime 1: $\bar{h}\lambda_*\geq 1$ and $\lambda\lambda_*>\alpha$}
\subsubsection{Calculation of  $M(\lambda,\lambda_*)$}
From now on we denote $D_{\mu_\epsilon}$ simply by $D_{\epsilon}$, $C_{\mu_\epsilon}$ by $C_\epsilon$, etc. Firstly, when $\bar{h}\lambda_*\geq 1$ and $\lambda\lambda_*>\alpha$, we get that
\begin{equation}
    \frac{\partial f_{0}(\lambda,\lambda_*)}{\partial\lambda_*} = \frac{\d g_{\lambda,\epsilon}^{(\alpha)}(x)}{\d x}\bigg|_{\epsilon=0,x=(D^{-1}_{0}(\frac{1}{\lambda_*}))^2} \frac{\partial(D^{-1}_{\epsilon}(\frac{1}{\lambda_*}))^2}{\partial\lambda_*}\bigg|_{\epsilon=0}. \label{derivative-fwrtlam}
\end{equation}
Using the relation linking the $D$-transform and rectangular $R$-transform, we can easily get that
\begin{IEEEeqnarray}{rl}
    &\frac{\partial(D^{-1}_{\epsilon}(\frac{1}{\lambda_*}))^2}{\partial\lambda_*}\bigg|_{\epsilon=0} = - \frac{1}{\lambda_*}C_0^\prime\bigg(\frac{1}{\lambda_*}\bigg)\bigg(2\alpha C_0\bigg(\frac{1}{\lambda_*}\bigg)+\alpha+1\bigg) \nn
    & \qquad\qquad\qquad\quad +\bigg(\alpha C_0\bigg(\frac{1}{\lambda_*}\bigg)+1\bigg)\bigg(C_0\bigg(\frac{1}{\lambda_*}\bigg)+1\bigg). \label{derivation-Dlamb}
\end{IEEEeqnarray}
The other term we need to compute is $\d_xg_{\lambda,\epsilon}^{(\alpha)}(x)$ when $\epsilon=0$. We can get that $\int\hat{\rho}(\d t)\frac{x}{x-t} = \alpha T^{(\alpha),-1}(xD(\sqrt{x}))+1$ by the fact that $xD(\sqrt{x}) = T^{(\alpha)}\big(\int\rho(\d t)\frac{x}{x-t}-1\big)$. Thus, the derivative of $g_{\lambda,\epsilon}^{(\alpha)}(x)$ with respect to $x$ can be written as
\begin{equation}
    \frac{\d g_{\lambda,\epsilon}^{(\alpha)}(x)}{\d x} = -\frac{1}{2x}\bigg(T^{(\alpha),-1}(xD_{\epsilon}(\sqrt{x}))-T^{(\alpha),-1}\bigg(\frac{\lambda x}{\alpha}\bigg) \bigg) .\label{derivative-gx}
\end{equation}
Combining (\ref{M-flambda}), (\ref{derivative-fwrtlam}), (\ref{derivation-Dlamb}) and (\ref{derivative-gx}), we can get that when $\bar{h}\lambda_*\geq 1$ and $\lambda\lambda_*>\alpha$ (note from Conjecture~\ref{Conj:LogZ} that the log-partition function depends on $\lambda_*$ only in the present temperature regime, so from \eqref{M-flambda} we know that $M(\lambda,\lambda_*)$ vanishes for the other regimes), $M(\lambda,\lambda_*)$ is in \eqref{eq:MQ-inf}.

\subsubsection{Calculation of  $Q(\lambda,\lambda_*)$}

Secondly, we will calculate the derivative of $f_{\epsilon}^{(\alpha)}(\lambda,\lambda_*)$ with respect to $\epsilon$. Notice that the three parts of $f_{\epsilon}^{(\alpha)}(\lambda,\lambda_*)$ are all related to $\epsilon$,  so our calculation also has three parts. When $\bar{h}\lambda_*\geq 1$ and $\lambda\lambda_*>\alpha$, we know that
\begin{IEEEeqnarray}{rCl}
    \frac{\partial f^{(\alpha)}_{\epsilon}(\lambda,\lambda_*)}{\partial\epsilon} &=& \frac{\d g_{\lambda,\epsilon}^{(\alpha)}(x)}{\d x}\bigg|_{\epsilon=0,x=(D^{-1}_{0}(\frac{1}{\lambda_*}))^2}\!\!\!\frac{\partial(D^{-1}_{\epsilon}(\frac{1}{\lambda_*}))^2}{\partial\epsilon}\bigg|_{\epsilon=0} \nn
    && +\> \frac{\partial g_{\lambda,\epsilon}^{(\alpha)}(x)}{\partial\epsilon}\bigg|_{\epsilon=0,x=(D^{-1}_{0}(\frac{1}{\lambda_*}))^2}. \label{eq:Q-regime1}
\end{IEEEeqnarray}
The first term is easy to get since
\begin{equation*}
    \frac{\partial(D^{-1}_{\epsilon}(\frac{1}{\lambda_*}))^2}{\partial\epsilon}\bigg|_{\epsilon=0} = 2\alpha C_{0}\bigg(\frac{1}{\lambda_*}\bigg)+\alpha+1.
\end{equation*}
Let us focus on the second term. We observe that the only term related to $\epsilon$ in $g_{\lambda,\epsilon}^{(\alpha)}(x)$ is
\begin{equation*}
    \ell(x,\epsilon) = -\frac{1}{2}\int\rho_{\epsilon}(\d t)\ln(x-t) + \text{const},
\end{equation*}
where we use the relation $\hat{\rho}_\epsilon=\alpha{\rho}_\epsilon+(1-\alpha)\delta_0$ and the last constant can be ignored since we just need to calculate the derivative of $\ell(x,\epsilon)$ with respect to $\epsilon$.

Let $\bm{H}(t)=\bm{A} +\frac{1}{\sqrt{m}}\bm{G}(t) \in \mathbb{R}^{n\times m}$, where $\bm{A}$ is deterministic and the elements of $\bm{G}(t)$ are independent Brownian motions. The dynamics of the eigenvalues $\lambda_1(t)\geq \lambda_2(t)\geq \dots \geq\lambda_n(t)$ of $\bm{H}(t)\bm{H}(t)^*$ has been intensively studied, called Laguerre or Wishart process. Specifically, Theorem 2.1 in \cite{guionnet2021large} gives that, for $1\leq i\leq n$,
\begin{equation}
    \d \lambda_i(t) = 2\sqrt{\lambda_i}\frac{\d B_i}{\sqrt{m}} + \bigg( \frac{1}{m}\sum_{j:j\neq i}^{n}\frac{\lambda_i+\lambda_j}{\lambda_i-\lambda_j}+1\bigg)\d t, \label{Wishart-process}
\end{equation}
where $B_1, B_2, \ldots, B_n$ are independent Brownian motions. Let
\begin{equation*}
    \ell_n(x,\{\lambda_i\}) = -\frac{1}{n}\sum_{i=1}^n\ln(x-\lambda_i).
\end{equation*}
Then, 
\begin{equation*}
    \frac{\partial\ell_n}{\partial\lambda_i}=\frac{1}{n}\frac{1}{x-\lambda_i},\;\;\;\; \frac{\partial^2\ell_n}{\partial\lambda_i^2}=\frac{1}{n}\frac{1}{(x-\lambda_i)^2}.
\end{equation*}
Using Ito's lemma (see (8.15)-(8.17) in \cite{potters2020first}), we have that
\begin{IEEEeqnarray}{rl}
    \!\dif \ell_n &= \sum_{i=1}^n\frac{\partial\ell_n}{\partial\lambda_i}\dif\lambda_i+\bigg(\sum_{i=1}^n\frac{2\lambda_i}{m}\frac{\partial^2\ell_n}{\partial\lambda_i^2}\bigg)\dif t \nn
    &= \frac{1}{n}\sum_{i=1}^{n}\frac{1}{x-\lambda_i}\dif t + \frac{1}{mn}\sum_{i,j:i\neq j}^n\frac{\lambda_i+\lambda_j}{(x-\lambda_i)(\lambda_i-\lambda_j)}\dif t \nn
    & \;\> +\>\frac{2}{mn}\sum_{i=1}^n\frac{\lambda_i}{(x-\lambda_i)^2}\dif t + \frac{4}{n\sqrt{m}}\sum_{i=1}^n\frac{\sqrt{\lambda_i}}{x-\lambda_i}\dif B_i. \label{dl-temp1}
\end{IEEEeqnarray}
Then, the second term can be simplified as
\begin{equation*}
    \frac{1}{2mn}\sum_{i,j=1}^n\frac{\lambda_i+\lambda_j}{(x-\lambda_i)(x-\lambda_j)}\dif t - \frac{1}{mn}\sum_{i=1}^n\frac{\lambda_i}{(x-\lambda_i)^2}\dif t.
\end{equation*}
We can easily calculate that
\begin{equation*}
    \frac{1}{n}\sum_{i=1}^n\frac{\lambda_i}{x-\lambda_i}=-1-x\frac{\partial \ell_n}{\partial x},\;\;\;\; \frac{1}{n}\sum_{i=1}^n\frac{1}{x-\lambda_i}=-\frac{\partial \ell_n}{\partial x}.
\end{equation*}
Plugging this equation into (\ref{dl-temp1}), we get that
\begin{IEEEeqnarray}{rCl}
    \dif \ell_n &=& \frac{4}{n\sqrt{m}}\sum_{i=1}^n\frac{\sqrt{\lambda_i}}{x-\lambda_i}\dif B_i + \frac{2}{mn}\sum_{i=1}^n\frac{\lambda_i}{(x-\lambda_i)^2}\dif t \nn
    && -\>\frac{\partial \ell_n}{\partial x}\dif t+\alpha\frac{\partial \ell_n}{\partial x}\Big(1+x\frac{\partial\ell_n}{\partial x}\Big)\dif t.
\end{IEEEeqnarray}
We now take the expectation on both sides of the equation, and notice that the first term is already zero. Then, we also notice that the second term will vanish when $n\to\infty$, so we get that
\begin{equation*}
    \mathbb{E}[\dif\ell_n]=\mathbb{E}\bigg[\alpha x\bigg(\frac{\partial\ell_n}{\partial x}\bigg)^2+(\alpha-1)\frac{\partial\ell_n}{\partial x}\bigg]\dif t + O\bigg(\frac{1}{n}\bigg).
\end{equation*}
Therefore, when $\ell=\lim_{n\to\infty}\ell_n$, we get that
\begin{equation}
    \frac{\partial\ell}{\partial t}=\alpha x\bigg(\frac{\partial\ell}{\partial x}\bigg)^2+(\alpha-1)\frac{\partial\ell}{\partial x}. \label{eq:deriv-lwrtt}
\end{equation}
Using this equation, we can calculate that
\begin{equation*}
    \frac{\partial\ell(x,\epsilon)}{\partial\epsilon}\bigg|_{\epsilon=0,x=(D^{-1}_{0}(\frac{1}{\lambda_*}))^2}\!\!\!\! = \frac{1}{2}D_0(\sqrt{x})\bigg|_{x=(D^{-1}_{0}(\frac{1}{\lambda_*}))^2}\!\!\!\! = \frac{1}{2\lambda_*}.
\end{equation*}
Therefore, we finally get the $Q(\lambda,\lambda_*)$ in \eqref{eq:MQ-inf} by taking the above equation into \eqref{eq:Q-regime1} when $\bar{h}\lambda_*\geq 1$ and $\lambda\lambda_*>\alpha$.

\subsection{Low temperature regime 2: $\bar{h}\lambda_*<1$ and $\lambda>\alpha\bar{h}$}
Next, when $\bar{h}\lambda_*<1$ and $\lambda>\alpha\bar{h}$, we can get that
\begin{IEEEeqnarray}{rCl}
    \frac{\partial f_{\epsilon}^{(\alpha)}(\lambda,\lambda_*)}{\partial\epsilon} &=& \frac{\d g_{\lambda,\epsilon}^{(\alpha)}(x)}{\d x}\bigg|_{\epsilon=0,x=\bar{\gamma}_{0}^2}\cdot \frac{\partial\bar{\gamma}_{\epsilon}^2}{\partial\epsilon}\bigg|_{\epsilon=0} \nn
    && +\> \frac{\partial g_{\lambda,\epsilon}^{(\alpha)}(x)}{\partial\epsilon}\bigg|_{\epsilon=0,x=\bar{\gamma}_{0}^2}. \label{eq:compu-mse-2}
\end{IEEEeqnarray}
The only term we need to calculate is $\partial_{\epsilon}\bar{\gamma}_{\epsilon}^2$. From (\ref{Wishart-process}) we know that
\begin{equation*}
    \frac{\dif\lambda_1}{\dif t}=\frac{1}{m}\sum_{j=2}^n\frac{\lambda_1+\lambda_j}{\lambda_1-\lambda_j}+1.
\end{equation*}
Thus, we can get that
\begin{IEEEeqnarray}{rl}
    \frac{\dif\bar{\gamma}_{\epsilon}^2}{\dif\epsilon}\bigg|_{\epsilon=0} &= \lim_{m\to\infty}\bigg(\frac{1}{m}\sum_{j=2}^n\frac{\lambda_1+\lambda_j}{\lambda_1-\lambda_j}+1\bigg) \nn
    &= \lim_{z\to\bar{\gamma}_0^2}\int\hat{\rho}(\dif t)\frac{z+t}{z-t}+\alpha \nn
    &=2\alpha T^{(\alpha),-1}(\bar{\gamma}_0^2D_0(\bar{\gamma}_0^+))+\alpha+1.\label{eq:laste}
\end{IEEEeqnarray}
Therefore, by combining \eqref{eq:compu-mse-2} and \eqref{eq:laste}, we obtain $Q(\lambda,\lambda_*)$ in \eqref{eq:MQ-inf} for this regime, i.e., when $\bar{h}\lambda_*<1$ and $\lambda>\alpha\bar{h}$, where we recall that $\bar{h}=\bar{h}_0=\lim_{z\downarrow\bar{\gamma}_0}D_{0}(z)$ and we already assume that it is finite.

\subsection{High temperature regime: otherwise \iffalse\teng{($\bar{h}\lambda_*\geq 1\cap\lambda\lambda_*\leq\alpha$ or $\bar{h}\lambda_*<1\cap\lambda\leq\alpha\bar{h}$)}\fi}
Finally, the last case is easy to calculate. Using the rectangular free convolution (see Theorem 3.12 and 3.13 of \cite{benaych2011rectangular}) in the generalized observation model and noticing that the rectangular R-transform of $\sqrt{\epsilon}\bm{W}$ is $\epsilon t$, we can get that $C_{\epsilon}^{(\alpha)}(t)=C_{0}^{(\alpha)}(t)+\epsilon t$. Thus,
\begin{equation}
    \frac{\partial f_{\epsilon}^{(\alpha)}(\lambda,\lambda_*)}{\partial\epsilon} = \int_{0}^{\sqrt{\frac{\lambda}{\alpha}}}\frac{t^2}{t}\dif t = \frac{\lambda}{2\alpha}.
\end{equation}
This gives that $Q(\lambda,\lambda_*)=0$ in the high temperature regime.

\section{Proofs for approximate message passing}\label{sec:appendix-seproof}

\subsection{Auxiliary AMP and its state evolution}
To prove the state evolution result for the Gaussian AMP, we define the so-called auxiliary AMP as follows:
\begin{equation}
\begin{IEEEeqnarraybox}{rClrCl}
    \tilde{\bm{z}}^t &=& \bm{Z}^{\mathsf{T}}\tilde{\bm{u}}^t-\sum_{i=1}^{t-1}b^{\text{A}}_{ti}\tilde{\bm{v}}^i,\;\; &\tilde{\bm{v}}^t &=& \tilde{v}_t(\tilde{\bm{z}}^1,\ldots,\tilde{\bm{z}}^t), \\
    \tilde{\bm{y}}^t &=& \bm{Z}\tilde{\bm{v}}^t-\sum_{i=1}^{t}a^{\text{A}}_{ti}\tilde{\bm{u}}^i,\;\; &\tilde{\bm{u}}^{t+1} &=& \tilde{u}_{t+1}(\tilde{\bm{y}}^1,\ldots,\tilde{\bm{y}}^t),
\end{IEEEeqnarraybox}
\label{eq:auxi-amp}
\end{equation}
where $\tilde{\bm{u}}^1=\bm{u}^1$ in \eqref{eq:gauss-amp}. The non-linear function $\tilde{v}_t$ and $\tilde{u}_{t+1}$ are defined recursively as
\begin{IEEEeqnarray}{rl}
    &\tilde{v}_t(z_1,\ldots,z_t,v_*)=v_t\Big(\!\!-\!\alpha\bar{\beta}_t\tilde{v}_{t-1}(z_1,\ldots,z_{t-1},v_*) \nn
    &\quad\; +\>z_t+\bar{\nu}_{t}v_*+\sum_{i=1}^{t-1}(\bar{\bm{B}}_t)_{t,i}\tilde{v}_{i}(z_1,\ldots,z_{i},v_*)\Big), \label{tildefunc1} \\
    &\tilde{u}_{t+1}(y_1,\ldots,y_t,u_1,u_*)=u_{t+1}\Big(\!\!-\!\bar{\alpha}_t\tilde{u}_{t}(y_1,\ldots,y_{t-1},u_1,u_*) \nn
    &\quad\; +\>y_t+\bar{\mu}_{t}u_*+\sum_{i=1}^{t}(\bar{\bm{A}}_t)_{t,i}\tilde{u}_{i}(y_1,\ldots,y_{i-1},u_1,u_*)\Big) ,\label{tildefunc2}
\end{IEEEeqnarray}
where $\tilde{u}_1(u_1,u_*):=u_1$. The parameters $(\bar{\nu}_t,\bar{\beta}_t,\bar{\bm{B}}_t)$ and $(\bar{\mu}_t,\bar{\alpha}_t,\bar{\bm{A}}_t)$ come from the true Gaussian AMP. We can obtain the coefficient $\{b^{\text{A}}_{t,i}\}_{i=1}^{t-1}$ and $\{a^{\text{A}}_{t,i}\}_{i=1}^{t}$ in the same way as \eqref{eq:corr-ABSigmaOmega} by replacing $\bar{\bm{\Delta}}_t$, $\bar{\bm{\Gamma}}_t$, $\bar{\bm{\Phi}}_t$, $\bar{\bm{\Psi}}_t$ with $\bm{\Delta}_t^{\text{A}}$, $\bm{\Gamma}_t^{\text{A}}$, $\bm{\Phi}_t^{\text{A}}$, $\bm{\Psi}_t^{\text{A}}$, which are defined as follows:
\begin{equation}
\begin{IEEEeqnarraybox}{rl}
    &({\bm{\Delta}}_t^{\text{A}})_{ij} = \langle\tilde{\bm{u}}^i\tilde{\bm{u}}^j\rangle,\;\; ({\bm{\Gamma}}_t^{\text{A}})_{ij} = \langle\tilde{\bm{v}}^i\tilde{\bm{v}}^j\rangle,\;\; 1\leq i,j\leq t, \\
    &({\bm{\Phi}}_t^{\text{A}})_{ij} = \langle\partial_j\tilde{\bm{u}}^i\rangle,\;\; 1\leq i<j\leq t, \\
    &({\bm{\Psi}}_t^{\text{A}})_{ij} = \langle\partial_j\tilde{\bm{v}}^i\rangle,\;\; 1\leq i\leq j\leq t.
\end{IEEEeqnarraybox}
\end{equation}
Here, $\tilde{\bm{u}}^i\tilde{\bm{u}}^j,\,\partial_j\tilde{\bm{u}}^i\in\mathbb{R}^n$ denote the entrywise product and partial derivative with respect to $z_j$, respectively, and $\langle\bm{u}\rangle=(\sum_{i=1}^nu_i)/n$ denotes the mean of a vector. Therefore, from \cite{fan2022approximate}, we can get the state evolution of the auxiliary AMP as follows:
\begin{equation}
\begin{IEEEeqnarraybox}{rl}
    &(\tilde{Z}_1,\ldots,\tilde{Z}_t)\sim\mathcal{N}(0,\bar{\bm{\Omega}}^{\text{A}}_t),\; \tilde{V}_t=\tilde{v}_t(\tilde{Z}_1,\ldots,\tilde{Z}_t,V_*), \\
    &(\tilde{Y}_1,\ldots,\tilde{Y}_t)\sim\mathcal{N}(0,\bar{\bm{\Sigma}}^{\text{A}}_t),\; \tilde{U}_{t+1}=\tilde{u}_{t+1}(\tilde{Y}_1,\ldots,\tilde{Y}_t,\tilde{U}_1,U_*),
\end{IEEEeqnarraybox}
\label{eq:se-of-auxi-amp}
\end{equation}
where the random variable $\tilde{U}_1=U_1$ in \eqref{eq:se-of-gauss-amp}. Furthermore, the coefficients can be given by
\begin{equation}
\begin{IEEEeqnarraybox}{rl}
    &(\bar{\bm{\Delta}}_t^{\text{A}})_{ij} = \mathbb{E}[\tilde{U}_i\tilde{U}_j],\;\; (\bar{\bm{\Gamma}}_t^{\text{A}})_{ij} = \mathbb{E}[\tilde{V}_i\tilde{V}_j],\;\; 1\leq i,j\leq t, \\
    &(\bar{\bm{\Phi}}_t^{\text{A}})_{ij} = \mathbb{E}[\partial_j\tilde{U}_i],\;\; 1\leq i<j\leq t, \\
    &(\bar{\bm{\Psi}}_t^{\text{A}})_{ij} = \mathbb{E}[\partial_j\tilde{V}_i],\;\; 1\leq i\leq j\leq t.
\end{IEEEeqnarraybox}
\end{equation}

We can define the matrices $\bar{\bm{\Omega}}^{\text{A}}_t$, $\bar{\bm{\Sigma}}^{\text{A}}_t$, $\bar{\bm{A}}^{\text{A}}_t$, $\bar{\bm{B}}^{\text{A}}_t$ through \eqref{eq:corr-ABSigmaOmega} by replacing $\bar{\bm{\Delta}}_t$, $\bar{\bm{\Gamma}}_t$, $\bar{\bm{\Phi}}_t$, $\bar{\bm{\Psi}}_t$ with $\bar{\bm{\Delta}}_t^{\text{A}}$, $\bar{\bm{\Gamma}}_t^{\text{A}}$, $\bar{\bm{\Phi}}_t^{\text{A}}$, $\bar{\bm{\Psi}}_t^{\text{A}}$. From \cite{fan2022approximate}, we obtain the following proposition about the state evolution of the auxiliary AMP we defined above.
\begin{proposition}{(State evolution of auxiliary AMP).} \label{proposition:se-of-auxi-amp}
Consider the auxiliary AMP in \eqref{eq:auxi-amp} and its state evolution in \eqref{eq:se-of-auxi-amp}. Let $\tilde{\psi}:\mathbb{R}^{t+2}\to\mathbb{R}$ and $\tilde{\phi}:\mathbb{R}^{t+1}\to\mathbb{R}$ be any pseudo-Lipschitz functions of order $2$. Then for each $t\geq 1$, we almost surely have
\begin{IEEEeqnarray}{rCl}
    &\lim_{m\to\infty}\frac{1}{m}\sum_{i=1}^m&\tilde{\phi}\big((\tilde{\bm{z}}^1)_i,\ldots,(\tilde{\bm{z}}^t)_i,(\tilde{\bm{v}}^1)_i,\ldots,(\tilde{\bm{v}}^t)_i,({\bm{v}}^*)_i\big) \nonumber \\
    &&=\mathbb{E}\tilde{\phi}\big(\tilde{Z}_1,\ldots,\tilde{Z}_t,\tilde{V}_1,\ldots,\tilde{V}_t,{V}_*\big), \\
    &\lim_{n\to\infty}\frac{1}{n}\sum_{i=1}^n&\tilde{\psi}\big((\tilde{\bm{y}}^1)_i,\ldots,(\tilde{\bm{y}}^t)_i,(\tilde{\bm{u}}^1)_i,\ldots,(\tilde{\bm{u}}^{t+1})_i,({\bm{u}}^*)_i\big) \nonumber \\
    &&=\mathbb{E}\tilde{\psi}\big(\tilde{Y}_1,\ldots,\tilde{Y}_t,\tilde{U}_1,\ldots,\tilde{U}_{t+1},{U}_*\big).
\end{IEEEeqnarray}
\end{proposition}

\subsection{Details of the proof}

We will use the auxiliary AMP and its state evolution to complete our proof. Our strategy is similar to the square case \cite{barbier2022price}, so we only provide a sketch. The first step is to show
\begin{equation}
    (\tilde{Z}_1,\ldots,\tilde{Z}_t)\overset{\text{d}}{=}(Z_1,\ldots,Z_t),\;\; (\tilde{Y}_1,\ldots,\tilde{Y}_t)\overset{\text{d}}{=}(Y_1,\ldots,Y_t),
    \label{eq:prove-first-step}
\end{equation}
where the random variables on the left are defined in \eqref{eq:se-of-auxi-amp} and the ones on the right are defined in \eqref{eq:se-of-gauss-amp}. Combining the definitions of the two state evolutions and of the non-linear functions $\tilde{v}_t$, $\tilde{u}_{t+1}$, \eqref{eq:prove-first-step} follows from an induction argument similar to that for the square case in \cite{barbier2022price}.

Then, the second step is to show that, for any pseudo-Lipschitz functions ${\psi}:\mathbb{R}^{t+2}\to\mathbb{R}$ and ${\phi}:\mathbb{R}^{t+1}\to\mathbb{R}$ of order $2$, the following limit holds almost surely for $t\geq 1$:
\begin{equation}
\begin{IEEEeqnarraybox}{rl}
    &\lim_{m\to\infty}\bigg|\frac{1}{m}\sum_{i=1}^m{\phi}\big((\tilde{\bm{g}}^1)_i,\ldots,(\tilde{\bm{g}}^t)_i,(\tilde{\bm{v}}^1)_i,\ldots,(\tilde{\bm{v}}^t)_i,(\tilde{\bm{v}}^*)_i\big) \\
    &-\frac{1}{m}\sum_{i=1}^m\phi\big((\bm{g}^1)_i,\ldots,(\bm{g}^t)_i,(\bm{v}^1)_i,\ldots,(\bm{v}^t)_i,(\bm{v}^*)_i\big)\bigg|=0,
\end{IEEEeqnarraybox}
\label{eq:prove-second-step1}
\end{equation}
\begin{equation}
\begin{IEEEeqnarraybox}{rl}
    &\lim_{n\to\infty}\bigg|\frac{1}{n}\sum_{i=1}^n{\psi}\big((\tilde{\bm{f}}^1)_i,\ldots,(\tilde{\bm{f}}^t)_i,(\tilde{\bm{u}}^1)_i,\ldots,(\tilde{\bm{u}}^{t+1})_i,(\tilde{\bm{u}}^*)_i\big) \\
    &-\frac{1}{n}\sum_{i=1}^n\psi\big((\bm{f}^1)_i,\ldots,(\bm{f}^t)_i,(\bm{u}^1)_i,\ldots,(\bm{u}^{t+1})_i,(\bm{u}^*)_i\big)\bigg|=0,
\end{IEEEeqnarraybox}
\label{eq:prove-second-step2}
\end{equation}
where we define for $s\in\{1,2,\ldots,t\}$,
\begin{IEEEeqnarray}{rl}
    \tilde{\bm{g}}^s &=\tilde{\bm{z}}^s+\bar{\nu}_s\bm{v}^*-\alpha\bar{\beta}_s\tilde{\bm{v}}^{s-1}+\sum_{i=1}^{s-1}(\bar{\bm{B}}_s)_{s,i}\tilde{\bm{v}}^i, \nn
    \tilde{\bm{f}}^s &=\tilde{\bm{y}}^s+\bar{\mu}_s\bm{u}^*-\bar{\alpha}_s\tilde{\bm{u}}^{s}+\sum_{i=1}^{s}(\bar{\bm{A}}_s)_{s,i}\tilde{\bm{u}}^i. \nonumber
\end{IEEEeqnarray}
This claim follows from an application of Cauchy-Schwarz inequality and another induction argument, which is also similar to the square case analyzed in \cite{barbier2022price}.

In the last step, by using Proposition \ref{proposition:se-of-auxi-amp} together with the conclusion of the second step \eqref{eq:prove-second-step1} and \eqref{eq:prove-second-step2}, we have that
\begin{IEEEeqnarray}{rCl}
    &\lim_{m\to\infty}\frac{1}{m}\sum_{i=1}^m&\phi\big((\bm{g}^1)_i,\ldots,(\bm{g}^t)_i,(\bm{v}^1)_i,\ldots,(\bm{v}^t)_i,(\bm{v}^*)_i\big) \nn
    &&=\mathbb{E}{\phi}\big(\tilde{G}_1,\ldots,\tilde{G}_t,\tilde{V}_1,\ldots,\tilde{V}_t,{V}_*\big), \label{eq:prove-last-1} \\
    &\lim_{n\to\infty}\frac{1}{n}\sum_{i=1}^n&\psi\big((\bm{f}^1)_i,\ldots,(\bm{f}^t)_i,(\bm{u}^1)_i,\ldots,(\bm{u}^{t+1})_i,(\bm{u}^*)_i\big) \nn
    &&=\mathbb{E}{\psi}\big(\tilde{F}_1,\ldots,\tilde{F}_t,\tilde{U}_1,\ldots,\tilde{U}_{t+1},{U}_*\big), \label{eq:prove-last-2}
\end{IEEEeqnarray}
where we have defined for $s\in\{1,\ldots,t\}$,
\begin{IEEEeqnarray}{rl}
    \tilde{G}_s &=\tilde{Z}_s+\bar{\nu}_s{V}_*-\alpha\bar{\beta}_s\tilde{V}_{s-1}+\sum_{i=1}^{s-1}(\bar{\bm{B}}_s)_{s,i}\tilde{V}_i, \nn
    \tilde{F}_s &=\tilde{Y}_s+\bar{\mu}_s{U}_*-\bar{\alpha}_s\tilde{U}_{s}+\sum_{i=1}^{s}(\bar{\bm{A}}_s)_{s,i}\tilde{U}_i. \nonumber
\end{IEEEeqnarray}
Finally, using \eqref{eq:prove-first-step}, we obtain
\begin{IEEEeqnarray}{rl}
    &\mathbb{E}{\phi}\big(G_1,\ldots,G_t,V_1,\ldots,V_t,V_*\big) \nn
    &\qquad\quad =\mathbb{E}{\phi}\big(\tilde{G}_1,\ldots,\tilde{G}_t,\tilde{V}_1,\ldots,\tilde{V}_{t},{V}_*\big), \label{eq:prove-last-3} \\
    &\mathbb{E}{\psi}\big({F}_1,\ldots,{F}_t,{U}_1,\ldots,{U}_{t+1},{U}_*\big) \nn
    &\qquad\quad =\mathbb{E}{\psi}\big(\tilde{F}_1,\ldots,\tilde{F}_t,\tilde{U}_1,\ldots,\tilde{U}_{t+1},{U}_*\big). \label{eq:prove-last-4}
\end{IEEEeqnarray}
Combining the above equalities \eqref{eq:prove-last-1}-\eqref{eq:prove-last-4}, we finally get \eqref{eq:se-of-gauss-amp}, which concludes the proof.

\end{document}